\documentclass[journal,draftcls,onecolumn,12pt,twoside]{IEEEtran}

\usepackage{times}
\usepackage{amsmath}
\usepackage{amssymb}
\usepackage{amsthm}
\usepackage{color}
\usepackage{algorithm}
\usepackage[noend]{algorithmic}
\usepackage{graphicx}
\usepackage{subfig}
\usepackage{multirow}
\usepackage[bookmarks=false,colorlinks=false,pdfborder={0 0 0}]{hyperref}
\usepackage{cite}
\usepackage{bm}
\usepackage{arydshln}
\usepackage{mathtools}
\usepackage{microtype}

\newtheorem{theorem}{Theorem}

\newtheorem{lemma}[theorem]{Lemma}

\long\def\symbolfootnote[#1]#2{\begingroup
\def\thefootnote{\fnsymbol{footnote}}\footnote[#1]{#2}\endgroup}

\title{Multi-Layer Transformed MDS Codes with Optimal Repair Access and Low Sub-Packetization}

\author{Hanxu Hou, Patrick P. C. Lee, and Yunghsiang S. Han
}

\begin{document}

\maketitle
\vspace{-0.5cm}
\begin{abstract}\symbolfootnote[0]{Hanxu Hou and Yunghsiang S. Han are with
the School of Electrical Engineering \& Intelligentization, Dongguan
University of Technology (E-mail: houhanxu@163.com, yunghsiangh@gmail.com).
Patrick P. C. Lee is with the Department of Computer Science and Engineering,
The Chinese University of Hong Kong (E-mail: pclee@cse.cuhk.edu.hk).  This
work was partially supported by the National Natural Science Foundation of
China (No.  61701115, 61671007), Start Fund of Dongguan University of
Technology (No. GB200902-19, KCYXM2017025), and Research Grants Council of
Hong Kong (GRF 14216316 and CRF C7036-15G).
}
An $(n,k)$ maximum distance separable (MDS) code has optimal repair access if the minimum number of symbols accessed from $d$ surviving nodes is achieved, where $k+1\le d\le n-1$. Existing results show that the sub-packetization $\alpha$ of an $(n,k,d)$ high code rate (i.e., $k/n>0.5$) MDS code with optimal repair access is at least $(d-k+1)^{\lceil\frac{n}{d-k+1}\rceil}$. In this paper, we propose a class of multi-layer transformed MDS codes such that the sub-packetization is $(d-k+1)^{\lceil\frac{n}{(d-k+1)\eta}\rceil}$, where $\eta=\lfloor\frac{n-k-1}{d-k}\rfloor$, and the repair access is optimal for any single node. We show that the
sub-packetization of the proposed multi-layer transformed MDS codes is
strictly less than the existing known lower bound when $\eta=\lfloor\frac{n-k-1}{d-k}\rfloor>1$, achieving by restricting the choice of $d$
specific helper nodes in repairing a failed node.  We further propose
multi-layer transformed EVENODD codes that have optimal repair access
for any single node and lower sub-packetization than the existing binary
MDS array codes with optimal repair access for any single node. With our
multi-layer transformation, we can design
new MDS codes that have the properties of low computational
complexity, optimal repair access for any single node, and relatively small
sub-packetization, all of which are critical for maintaining the reliability
of distributed storage systems.
\end{abstract}


\IEEEpeerreviewmaketitle

\section{Introduction}

Erasure codes are now widely adopted in distributed storage systems (e.g.,
\cite{ford2010}) that can provide significantly higher fault tolerance and
lower storage redundancy than traditional replication.  {\em Maximum distance
separable (MDS)} codes are a class of the erasure codes that can offer
reliability with the minimum amount of storage redundancy.  Specifically, an
$(n,k)$ MDS code, where $n$ and $k<n$ are two configurable parameters, encodes
a data file of $k$ {\em symbols} (i.e., the units for erasure code operations)
is encoded into $n$ symbols that are distributed in $n$ nodes, such that the
data file can be retrieved from any $k$ out of $n$ nodes; in the meantime, it
achieves the minimum redundancy $n/k$.

Upon the failure of any storage node, we want to repair the lost data in the
failed node by downloading the minimum possible amounts of symbols from other
surviving nodes.  \emph{Regenerating codes} (RGCs) \cite{dimakis2010} are a
special class of erasure codes that provably minimize the
{\em repair bandwidth}, defined as the total number of symbols being
downloaded for repairing a failed node.  In particular, minimum storage
regenerating (MSR) codes are a sub-class of RGCs that correspond to the
minimum storage point of RGCs.  MSR codes can also be viewed as a special
class of MDS codes that achieve the minimum repair bandwidth for repairing any
single failed node.

\subsection{Background of Optimal Repair Access}

We present the background of achieving the minimum repair bandwidth as
follows.  Specifically, we construct an $(n,k)$ MDS code through
{\em sub-packetization}, in which we encode a data file of $k\alpha$ symbols
(where $\alpha\ge 1$) over the finite field $\mathbb{F}_q$ into $n\alpha$
symbols that are distributed across $n$ nodes, each of which stores $\alpha$
symbols, such that the data file can be retrieved from any $k$ out of $n$
nodes.  Here, the number of symbols stored in each node, $\alpha$, is called
the {\em sub-packetization level}.

If any storage node fails, a repair operation is triggered to repair the
$\alpha$ lost symbols in a new storage node by downloading the available
symbols from $d$ surviving nodes (called {\em helper nodes}), where $k\le d\le
n-1$.  Dimakis {\em et al.} \cite{dimakis2010} show that the number of symbols
downloaded from each of the $d$ helper nodes, denoted by $\beta$, for
repairing the $\alpha$ lost symbols of a failed node in any $(n,k)$ MDS code
is lower-bounded by
\[
\beta=\frac{\alpha}{d-k+1}.
\]
Thus, the minimum repair bandwidth of any $(n,k)$ MDS code, denoted by
$\gamma$, is:
\begin{equation}
\gamma=d\beta=\frac{d\alpha}{d-k+1}.
\label{eq:optimal-repair}
\end{equation}

MSR codes \cite{dimakis2010} are a special class of MDS codes that achieve the
minimum repair bandwidth in Equation~\eqref{eq:optimal-repair}.  Furthermore,
we define {\em repair access} as the total number of symbols accessed (i.e.,
I/Os) from the $d$ helper nodes during the repair of a failed node.  We can
readily show that the minimum repair access of an MSR code is no less than the
minimum repair bandwidth.  We say that an MDS code has the
{\em optimal repair access} (a.k.a. {\em repair-by-transfer} \cite{shah2012b}) if the
minimum repair access is equal to the minimum repair bandwidth in
Equation~\eqref{eq:optimal-repair} is achieved; in other words, the minimum
amount of I/Os for repair is equal to the minimum repair bandwidth.
%MSR codes with optimal repair access are a special class of MSR codes that
%achieve the optimal repair access in Equation~\eqref{eq:optimal-repair}.

\subsection{Related Work}

Many constructions of MSR codes
\cite{rashmi2011,hou2016,suh2011,tamo2013,li2017,goparaju2017,ye2017,ye2017explicit1}
have been proposed in the literature.  For example, product-matrix MSR codes
\cite{rashmi2011} support the parameters that satisfy $2k-2\leq d\leq n-1$,
and are subsequently extended with lower computational complexity
\cite{hou2016}.  Another construction of MSR codes is based on interference
alignment \cite{suh2011}.  However, the above two constructions of MSR codes
are only suitable for low code rates (i.e, $k/n\leq 0.5$).

MSR codes with high code rates (i.e, $k/n> 0.5$) are important in practice.
Some existing constructions of high-code-rate MSR codes are found in
\cite{tamo2013,li2017,goparaju2017,ye2017,ye2017explicit1}.  It is shown in
\cite{balaji2017} that a tight lower bound of the sub-packetization level of
high-code-rate MSR codes with optimal repair access is $(d-k+1)^{\lceil
\frac{n}{d-k+1}\rceil}$.  More generally, for any $(n,k)$ MDS code with
optimal repair access for each of $w$ nodes (where $w<n$), the minimum
sub-packetization level is $\alpha=(d-k+1)^{\lceil \frac{w}{d-k+1}\rceil}$.

There are other practical concerns in distributed storage systems, such as how
to mitigate the computational complexity.  Binary MDS array codes are a
special class of MDS codes that have low computational complexity, since the
encoding and decoding procedures only involve XOR operations. Typical examples
of binary MDS array codes are EVENODD \cite{blaum1995,blaum2001}, X-code
\cite{xu1999} and RDP \cite{corbett2004,blaum2006}. Some efficient decoding
methods of binary MDS array codes are given in
\cite{hou2018,huang2016,hou2018new,hou2018a,hou2018d}.  There have been also
many studies \cite{hou2017,hou2018c,hou2018b,hou2018,li2019} on the optimal
repair bandwidth of binary MDS array codes.

\subsection{Our Contributions}

In this paper, we propose a class of multi-layer transformed MDS codes with
any parameters $k+1\le d\le n-1$, such that the optimal repair access is
achievable for any single node while the sub-packetization level is less than
the lower bound $(d-k+1)^{\lceil \frac{n}{d-k+1}\rceil}$ when $\lfloor
\frac{n-k-1}{d-k}\rfloor >1$.  As a case study, we propose a class of
multi-layer transformed EVENODD codes with optimal repair access for any
single failed column, low sub-packetization level, and low computational
complexity.

The main contributions of this paper are as follows.
\begin{enumerate}
\item
We present a generic transformation for any $(n,k)$ MDS code to obtain another
$(n,k)$ MDS code that has optimal repair access for each of the chosen
$(d-k+1)\eta$ nodes, while the sub-packetization level of the transformed code
is $d-k+1$, where $\eta =\lfloor \frac{n-k-1}{d-k}\rfloor \geq 1$.
\item
By applying the proposed transformation for an $(n,k)$ MDS code
$\lceil\frac{n}{(d-k+1)\eta}\rceil$ times, we can obtain an $(n,k)$
multi-layer transformed MDS code that achieves optimal repair access for
repairing any single failed node, while the sub-packetization level is
$(d-k+1)^{\lceil \frac{n}{(d-k+1)\eta}\rceil}$. When
$\lfloor\frac{n-k-1}{d-k}\rfloor >1$, we show that the proposed multi-layer
transformed MDS code has a lower sub-packetization level than the lower bound
given in \cite{balaji2017} and also less than that of the existing
high-code-rate MSR codes with optimal repair access.
%\item The essential reason of our multi-layer transformed MDS codes with lower sub-packetization is as follows. We can enable optimal repair access for each of the chosen $(d-k+1)\eta$ nodes with the proposed generic transformation, while the existing transformations \cite{li2017,hou2018,li2019} can only enable optimal repair access for each of the chosen $d-k+1$ nodes. Therefore, our multi-layer transformed MDS codes are obtained by applying the proposed generic transformation for $\lceil \frac{n}{(d-k+1)\eta}\rceil$ times, while the MDS codes \cite{li2017,hou2018,li2019} with optimal repair access for any single node are obtained by applying the transformation for $\lceil \frac{n}{(d-k+1)}\rceil$ times.
%By applying the proposed transformation for an $(n,k)$ MDS code for $\lceil \frac{w}{(d-k+1)\eta}\rceil$ times, we can obtain an $(n,k)$ MDS code that has optimal repair access for $w$ nodes and the sub-packetization level is $(d-k+1)^{\lceil \frac{w}{(d-k+1)\eta}\rceil}$, which is less than the lower bound in \cite{balaji2017} when $\lfloor \frac{n-k-1}{d-k}\rfloor >1$.
\item
We present a binary array version of transformation for binary MDS array
codes. We use EVENODD codes as a motivating example to enable optimal repair
access for each of the chosen $(d-k+1)\eta$ columns. By recursively applying
the transformation for EVENODD codes $\lceil \frac{n}{(d-k+1)\eta}\rceil$
times, the obtained multi-layer transformed EVENODD codes have optimal repair
access for any single failed column and have lower sub-packetization level than that
of the existing binary MDS array codes with optimal repair access.
\end{enumerate}

Compared to the existing constructions \cite{li2017,ye2017,ye2017explicit1}
of MDS codes with optimal repair access, the main reason that our multi-layer
transformed MDS codes achieve lower sub-packetization level is that we restrict the
choice of $d$ specific helper nodes in repairing a failed node.  This enables
us to design a generic transformation for MDS codes to achieve optimal repair
access for each of the chosen $(d-k+1)\eta$ nodes.  We can design the
multi-layer transformed MDS codes by applying the proposed generic
transformation for any MDS codes $\lceil \frac{n}{(d-k+1)\eta}\rceil$ times,
such that the sub-packetization level is
$(d-k+1)^{\lceil \frac{n}{(d-k+1)\eta}\rceil}$ and the repair access is
optimal for any single node.  Our multi-layer transformed MDS codes also
require that the underlying field size to be sufficiently large, so as to
maintain the MDS property.

There exist some similar transformations \cite{li2017,hou2018,li2019} for MDS
codes to enable optimal repair bandwidth of some nodes.  The main differences
between our transformation and the transformations in
\cite{li2017,hou2018,li2019} are summarized as follows.  First, Li
\emph{et al.} \cite{li2017} propose a transformation for MDS codes to enable
optimal repair bandwidth for any of the chosen $r$ nodes. Our transformation
enables optimal repair bandwidth for any of the chosen $(d-k+1)\eta$ nodes.
The transformation in \cite{li2017} can be viewed as a special case of our
transformation with $\eta=1$ and $d=n-1$. Second, the transformations in
\cite{hou2018,li2019} can be viewed as the variants of the transformation in
\cite{li2017} that are designed for binary MDS array codes. With a slightly
modification, our transformation is also applicable to binary MDS array codes
and the transformation in \cite{hou2018,li2019} can be viewed as a special
case of our binary version of transformation.

The rest of the paper is organized as follows.
Section~\ref{sec:example} presents motivating examples that show the main
ideas of the multi-layer transformed MDS codes.
Section~\ref{sec:trans} presents the transformation that can be applied to any
MDS code to enable optimal repair access for each of the chosen $(d-k+1)\eta$
nodes.
Section~\ref{sec:const} shows the construction of the multi-layer transformed
MDS codes with optimal repair access for any single node by recursively
applying the proposed transformation for any MDS code given in Section
\ref{sec:trans}.
Section~\ref{sec:trans-binary} shows that we can obtain the multi-layer
transformed EVENODD codes that have optimal repair access for any single
column by recursively applying the transformation for EVENODD codes.
Section~\ref{sec:com} shows that the obtained multi-layer transformed MDS
codes with optimal repair access have lower sub-packetization level than that of the
existing MDS codes with optimal repair access.
Section~\ref{sec:conclu} concludes the paper.

\section{Motivating Examples}
\label{sec:example}

In this section, we present two motivating examples with $n=8$, $k=5$, and
$d=6$. The first example is obtained by applying the transformation given in
Section~\ref{sec:trans} for the first $(d-k+1) \lfloor
\frac{n-k-1}{d-k}\rfloor=4$ nodes, so as to achieve optimal repair access for
each of the first four nodes with the sub-packetization level
$\alpha=d-k+1=2$. The second example shows the multi-layer transformed code in
Section~\ref{sec:const} to achieve optimal repair access for each node with
the sub-packetization level $\alpha=(d-k+1)^{\left\lceil \frac{n}{(d-k+1)\lfloor
\frac{n-k-1}{d-k}\rfloor}\right\rceil}=4$. Note that the second example is
obtained by recursively applying the transformation given in
Section~\ref{sec:trans} twice.

\subsection{First Example}

Suppose that a data file contains 10 data symbols that are denoted by
$$a^{1}_{1},a^{1}_{2},a^{1}_{3},a^{1}_{4},a^{1}_{5},a^{2}_{1},a^{2}_{2},a^{2}_{3},a^{2}_{4},a^{2}_{5}$$
over the finite field $\mathbb{F}_q$.  We first compute six coded symbols
$a^1_6,a^1_7,a^1_8,a^2_6,a^2_7,a^2_8$ by
\begin{align*}
\begin{bmatrix}
a^1_6 & a^1_7 &a^1_8\\
a^2_6 & a^2_7 &a^2_8
\end{bmatrix}=\begin{bmatrix}
a^1_1 & a^1_2 &a^1_3 &a^1_4 &a^1_5 \\
a^2_1 & a^2_2 &a^2_3 &a^2_4 &a^2_5
\end{bmatrix}\begin{bmatrix}
1 & p_1 & p_1^2 \\
1 & p_2 & p_2^2 \\
1 & p_3 & p_3^2 \\
1 & p_4 & p_4^2 \\
1 & p_5 & p_5^2 \\
\end{bmatrix},
\end{align*}
where $p_1,p_2,p_3,p_4,p_5$ are distinct and non-zero elements in $\mathbb{F}_q$.

\begin{table*}[!t]
%\scriptsize
\caption{The storage with $n=8$, $k=5$ and $d=6$ by applying the transformation for the first four nodes, where $e_1,e_2$ are field elements in $\mathbb{F}_q$ except zero and one.}
%\vspace{-12pt}
\begin{center}
\begin{tabular}{|c|c|c|c|c|c|c|c|}
\hline
Node 1 & Node 2  & Node 3  & Node 4 & Node 5  & Node 6  & Node 7  & Node 8 \\
\hline
$a^1_{1}$& $a^1_{2}+a^2_{1}$ & $a^1_{3}$& $a^1_{4}+a^2_{3}$ & $a^1_{5}$ & $a^1_{6}$ & $a^1_{7}$ & $a^1_{8}$ \\
\hline
$a^2_{1}+e_1a^1_2$& $a^2_{2}$ &$a^2_{3}+e_2c^1_4$& $a^2_{4}$ & $a^2_{5}$ & $a^2_{6}$ & $a^2_{7}$ & $a^2_{8}$ \\
\hline
\end{tabular}
\end{center}
\label{table:moti-example}
\end{table*}

Table \ref{table:moti-example} shows the storage by applying the transformation given in Section \ref{sec:trans} for the first four nodes, where
$e_1,e_2$ are field elements in $\mathbb{F}_q$ except zero and one.
We can repair the two symbols in each of the first four nodes of the code in
Table \ref{table:moti-example} by accessing one symbol from each of the chosen
$d=6$ nodes, and the repair access is optimal. Suppose that node 1 fails. We
can recover the two symbols stored in node 1 by downloading the following six
symbols.
\begin{align*}
a^1_2+a^2_1,a^1_3,a^1_5,a^1_6,a^1_7,a^1_8.
\end{align*}
Specifically, we can first compute $a^1_1,a^1_2,a^1_4$ from $a^1_3,a^1_5,a^1_6,a^1_7,a^1_8$ by
\begin{align*}
\begin{bmatrix}
a^1_1 & a^1_2 & a^1_4\\
\end{bmatrix}=\begin{bmatrix}
a^1_6-a^1_3-a^1_5 &
a^1_7-p_3a^1_3-p_5a^1_5 &
a^1_8-p_3^2a^1_3-p_5^2a^1_5\\
\end{bmatrix}\cdot \begin{bmatrix}
1 & p_1 & p_1^2\\
1 & p_2 & p_2^2\\
1 & p_4 & p_4^2\\
\end{bmatrix}^{-1}.
\end{align*}
Then, we can recover $a^2_1+e_1a^1_2$ by $a^2_1+e_1a^1_2=(a^1_2+a^2_1)+(e_1-1)a^1_2$.
We can also repair node~3 by downloading the following symbols
\begin{align*}
a^1_1,a^1_4+a^2_3,a^1_5,a^1_6,a^1_7,a^1_8,
\end{align*}
from nodes $1,4,5,6,7,8$.
We can repair node 2 and node 4 by downloading the symbols
\begin{align*}
a^2_1+e_1a^1_2,a^2_4,a^2_5,a^2_6,a^2_7,a^2_8,
\end{align*}
from nodes $1,4,5,6,7,8$ and
\begin{align*}
a^2_2,a^2_3+e_2a^1_4,a^2_5,a^2_6,a^2_7,a^2_8+a^4_7,
\end{align*}
from nodes $2,3,5,6,7,8$, respectively.

\subsection{Second Example}
\label{sec:sec-exam}
Suppose now that a data file contains 20 data symbols that are denoted as
$c^{\ell}_{1},c^{\ell}_{2},c^{\ell}_{3},c^{\ell}_{4},c^{\ell}_{5}$ over finite
field $\mathbb{F}_q$ with $\ell=1,2,3,4$. In the example, we have eight nodes
with each node storing four symbols.
The code of the second example can be obtained by applying the transformation in Section \ref{sec:trans} for the last four nodes of the first example. The construction of the second example is given as follows.
We first generate two instances of the first example. Specifically, we
first compute 12 coded symbols $c^\ell_6,c^\ell_7,c^\ell_8$ with $\ell=1,2,3,4$ by
\begin{align*}
\begin{bmatrix}
c^\ell_6 & c^\ell_7 &c^\ell_8
\end{bmatrix}=\begin{bmatrix}
c^\ell_1 & c^\ell_2 &c^\ell_3 &c^\ell_4 &c^\ell_5
\end{bmatrix}\begin{bmatrix}
1 & p_1 & p_1^2 \\
1 & p_2 & p_2^2 \\
1 & p_3 & p_3^2 \\
1 & p_4 & p_4^2 \\
1 & p_5 & p_5^2 \\
\end{bmatrix},
\end{align*}
where $p_1,p_2,p_3,p_4,p_5$ are distinct and non-zero elements in $\mathbb{F}_q$, and then obtain the two instances of the first example as
\begin{equation}
\begin{bmatrix}
c^1_1 & c^1_2+c^2_1 & c^1_3 & c^1_4+c^2_3 & c^1_5 & c^1_6 & c^1_7 & c^1_8 \\
c^2_1+e_1c^1_2 & c^2_2 & c^2_3+e_2c^1_4 & c^2_4 & c^2_5 & c^2_6 & c^2_7 & c^2_8 \\
c^3_1 & c^3_2+c^4_1 & c^3_3 & c^3_4+c^4_3 & c^3_5 & c^3_6 & c^3_7 & c^3_8 \\
c^4_1+e_1c^3_2 & c^4_2 & c^4_3+e_2c^3_4 & c^4_4 & c^4_5 & c^4_6 & c^4_7 & c^4_8 \\
\end{bmatrix}.
\label{eq:two-ins}
\end{equation}

Table~\ref{table:A3} shows the example of the multi-layer transformed code in
Section~\ref{sec:const}, obtained by recursively applying the transformation
given in Section~\ref{sec:trans}  twice, where $e_1,e_2,e_3,e_4$ are field
elements in $\mathbb{F}_q$ except zero and one. The multi-layer transformed
code in Table \ref{table:A3} has optimal repair access for each of all eight
nodes. In the following, we show the detailed repair method for the
multi-layer transformed code in Table \ref{table:A3}.

\begin{table*}[!t]
%\scriptsize
\caption{The multi-layer transformed code with $n=8$, $k=5$ and $d=6$ by recursively applying the transformation twice, where $e_1,e_2,e_3,e_4$ are field elements in $\mathbb{F}_q$ except zero and one.}
%\vspace{-12pt}
\begin{center}
\begin{tabular}{|c|c|c|c|c|c|c|c|}
\hline
Node 1 & Node 2  & Node 3  & Node 4 & Node 5  & Node 6  & Node 7  & Node 8 \\
\hline
$c^1_{1}$& $c^1_{2}+c^2_{1}$ & $c^1_{3}$& $c^1_{4}+c^2_{3}$ & $c^1_{5}$ & $c^1_{6}+c^3_{5}$ & $c^1_{7}$ & $c^1_{8}+c^3_{7}$ \\
\hline
$c^2_{1}+e_1c^1_2$& $c^2_{2}$ &$c^2_{3}+e_2c^1_4$& $c^2_{4}$ & $c^2_{5}$ & $c^2_{6}+c^4_{5}$ & $c^2_{7}$ & $c^2_{8}+c^4_{7}$ \\
\hline
$c^3_{1}$& $c^3_{2}+c^4_{1}$ & $c^3_{3}$& $c^3_{4}+c^4_{3}$ & $c^3_{5}+e_3c^1_{6}$ & $c^3_{6}$ & $c^3_{7}+e_4c^1_{8}$ & $c^3_{8}$ \\
\hline
$c^4_{1}+e_1c^3_2$& $c^4_{2}$ &$c^4_{3}+e_2c^3_4$& $c^4_{4}$ & $c^4_{5}+e_3c^2_{6}$ & $c^4_{6}$ & $c^4_{7}+e_4c^2_{8}$ & $c^4_{8}$ \\
\hline
\end{tabular}
\end{center}
\label{table:A3}
\end{table*}

We demonstrate that we can repair the four symbols in each node of the multi-layer transformed code in Table \ref{table:A3} by accessing two symbols from each
of the chosen $d=6$ nodes. Suppose that node 1 fails. We can recover the four symbols stored in node 1 by downloading the first symbol and the third symbol from nodes $2,3,5,6,7,8$, i.e., by downloading the following symbols.
\begin{align*}
&c^1_2+c^2_1,c^1_3,c^1_5,c^1_6+c^3_5,c^1_7,c^1_8+c^3_7,\\
&c^3_2+c^4_1,c^3_3,c^3_5+e_3c^1_6,c^3_6,c^3_7+e_4c^1_8,c^3_8.
\end{align*}

\begin{figure*}
\centering
\includegraphics[width=0.79\textwidth]{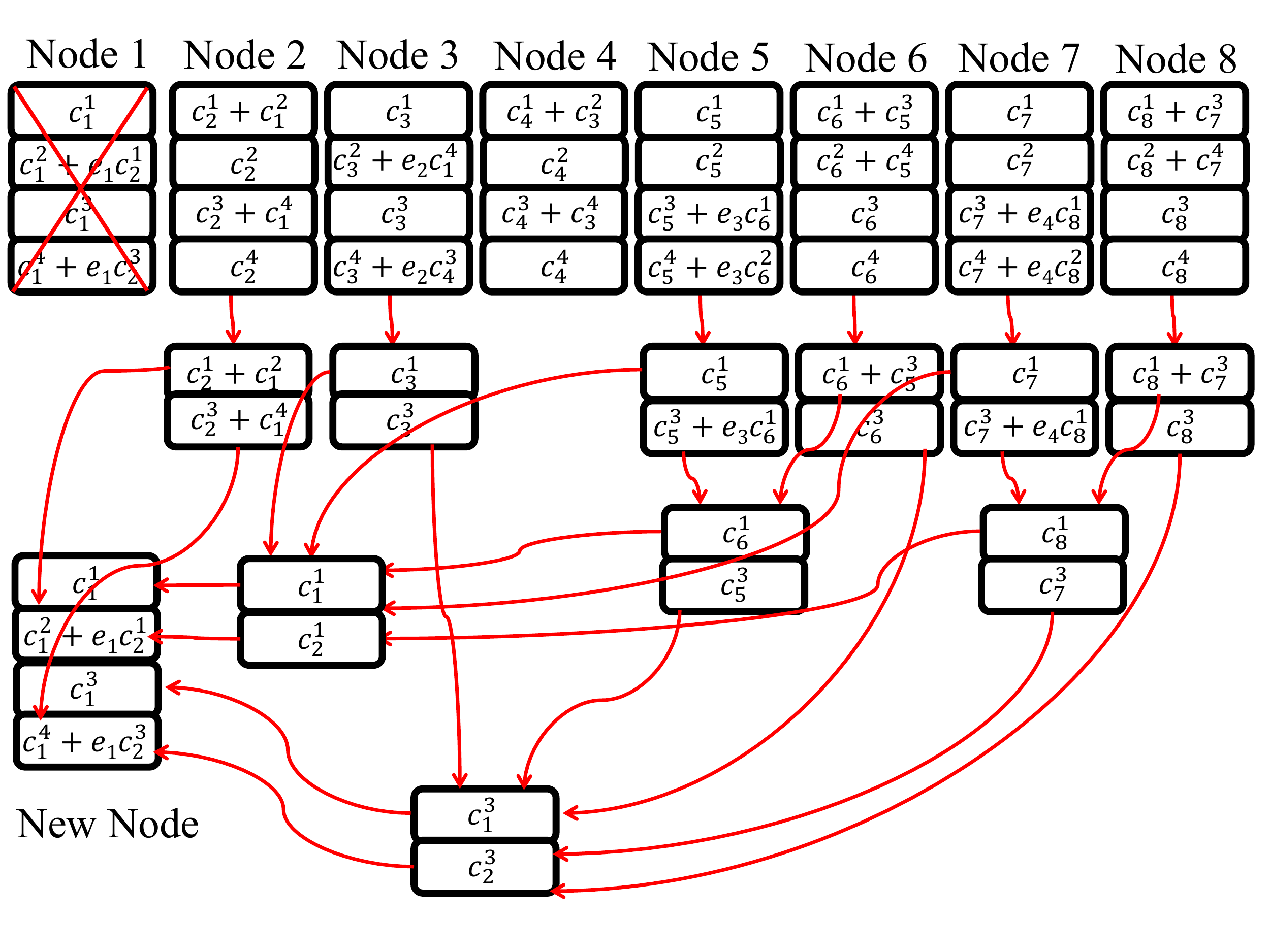}
\caption{The repair procedure of node 1 of the example given in Section~\ref{sec:example} with $k=5$, $r=3$ and $d=6$}
\label{example}
%\vspace{-0.5cm}
\end{figure*}

Fig. \ref{example} shows the detailed repair procedure of node 1. Specifically, we can first compute $c^1_6$ and $c^3_5$ from $c^1_6+c^3_5$ and $c^3_5+e_3c^1_6$, and compute $c^1_8$ and $c^3_7$ from $c^1_8+c^3_7$ and $c^3_7+e_4c^1_8$, as $e_3\neq 1$ and $e_4\neq 1$. Then, we can obtain $c^1_1,c^1_2,c^1_4$ and $c^3_1,c^3_2,c^3_4$ from $c^1_3,c^1_5,c^1_6,c^1_7,c^1_8$ and $c^3_3,c^3_5,c^3_6,c^3_7,c^3_8$ by
\begin{align*}
\begin{bmatrix}
c^1_1 & c^1_2 & c^1_4\\
\end{bmatrix}=\begin{bmatrix}
c^1_6-c^1_3-c^1_5 &
c^1_7-p_3c^1_3-p_5c^1_5 &
c^1_8-p_3^2c^1_3-p_5^2c^1_5\\
\end{bmatrix}\cdot \begin{bmatrix}
1 & p_1 & p_1^2\\
1 & p_2 & p_2^2\\
1 & p_4 & p_4^2\\
\end{bmatrix}^{-1},
\end{align*}
and
\begin{align*}
\begin{bmatrix}
c^3_1 & c^3_2 & c^3_4\\
\end{bmatrix}=\begin{bmatrix}
c^3_6-c^3_3-c^3_5 &
c^3_7-p_3c^3_3-p_5c^3_5 &
c^3_8-p_3^2c^3_3-p_5^2c^3_5\\
\end{bmatrix}\cdot \begin{bmatrix}
1 & p_1 & p_1^2\\
1 & p_2 & p_2^2\\
1 & p_4 & p_4^2\\
\end{bmatrix}^{-1},
\end{align*}
respectively.
Finally, we can recover $c^2_1+e_1c^1_2$ and $c^4_1+e_1c^3_2$ by $c^2_1+e_1c^1_2=(c^1_2+c^2_1)+(e_1-1)c^1_2$ and $c^4_1+e_1c^3_2=(c^3_2+c^4_1)+(e_1-1)c^3_2$, respectively.
We can also repair node 3 by downloading the following symbols
\begin{align*}
&c^1_1,c^1_4+c^2_3,c^1_5,c^1_6+c^3_5,c^1_7,c^1_8+c^3_7,\\
&c^3_1,c^3_4+c^4_3,c^3_5+e_3c^1_6,c^3_6,c^3_7+e_4c^1_8,c^3_8,
\end{align*}
from nodes $1,4,5,6,7,8$.
Node 2 and node 4 can be repaired by downloading the symbols
\begin{align*}
&c^2_1+e_1c^1_2,c^2_4,c^2_5,c^2_6+c^4_5,c^2_7,c^2_8+c^4_7,\\
&c^4_1+e_1c^3_2,c^4_4,c^4_5+e_3c^2_6,c^4_6,c^4_7+e_4c^2_8,c^4_8,
\end{align*}
from nodes $1,4,5,6,7,8$ and
\begin{align*}
&c^2_2,c^2_3+e_2c^1_4,c^2_5,c^2_6+c^4_5,c^2_7,c^2_8+c^4_7,\\
&c^4_2,c^4_3+e_2c^3_4,c^4_5+e_3c^2_6,c^4_6,c^4_7+e_4c^2_8,c^4_8,
\end{align*}
from nodes $2,3,5,6,7,8$, respectively.

Similarly, we can repair node 5 and node 7 by downloading the symbols
\begin{align*}
&c^1_1,c^1_2+c^2_1,c^1_3,c^1_4+c^2_3,c^1_6+c^3_5,c^1_7,\\
&c^2_1+e_1c^1_2,c^2_2,c^2_3+e_2c^4_1,c^2_4,c^2_6+c^4_5,c^2_7,
\end{align*}
from nodes $1,2,3,4,6,7$ and
\begin{align*}
&c^1_1,c^1_2+c^2_1,c^1_3,c^1_4+c^2_3,c^1_5,c^1_8+c^3_7,\\
&c^2_1+e_1c^1_2,c^2_2,c^2_3+e_2c^4_1,c^2_4,c^2_5,c^2_8+c^4_7,
\end{align*}
from nodes $1,2,3,4,5,8$, respectively, and repair node 6 and node 8 by downloading the symbols
\begin{align*}
&c^3_1,c^3_2+c^4_1,c^3_3,c^3_4+c^4_3,c^3_5+e_3c^1_6,c^3_8,\\
&c^4_1+e_1c^3_2,c^4_2,c^4_3+e_2c^3_4,c^4_4,c^4_5+e_3c^2_6,c^4_8,
\end{align*}
from nodes $1,2,3,4,5,8$ and
\begin{align*}
&c^3_1,c^3_2+c^4_1,c^3_3,c^3_4+c^4_3,c^3_6,c^3_7+e_4c^1_8,\\
&c^4_1+e_1c^3_2,c^4_2,c^4_3+e_2c^3_4,c^4_4,c^4_6,c^4_7+e_4c^2_8,
\end{align*}
from nodes $1,2,3,4,6,7$, respectively.

If the field size $q$ is large enough, we can always find the field elements $p_1,p_2,p_3,p_4,p_5,e_1,e_2,e_3,e_4$ in $\mathbb{F}_q$ such that any five out of the eight nodes can retrieve all 20 data symbols by Theorem \ref{thm:hybridcons} given in Section \ref{sec:mds}.

\section{A Generic Transformation}
\label{sec:trans}
In this section, we propose a generic transformation for MDS codes that can convert any MDS code into another MDS code that has optimal repair access for each of the chosen nodes and possesses low sub-packetization level. The multi-layer transformed MDS codes with optimal repair access for any single node given in the next section are obtained by recursively applying the transformation given in this section.

An $(n,k)$ MDS code encodes $k$ \emph{data symbols} $s_1,s_2,\ldots,s_k$ over finite field
$\mathbb{F}_q$ into $n$ \emph{coded symbols} $c_1,c_2,\ldots,c_n$ by
\[
\begin{bmatrix}
c_1 & c_2 & \cdots & c_{n}
\end{bmatrix}=
\begin{bmatrix}
s_1 & s_2 & \cdots & s_{k}
\end{bmatrix}\cdot \mathbf{G}_{k\times n},
\]
where $\mathbf{G}_{k\times n}$ is a $k\times n$ matrix. Any $k\times k$ sub-matrix of $\mathbf{G}_{k\times n}$
must be non-singular, in order to maintain the MDS property.
Typically, we can choose the matrix $\mathbf{G}_{k\times n}$
to be a $k\times n$ Vandermonde matrix or Cauchy matrix.
If we want to generate the systematic version of the code, i.e., $c_i=s_i$ for $i=1,2,\ldots,k$,
then the matrix $\mathbf{G}_{k\times n}$ is composed of a $k\times k$ identity matrix $\mathbf{I}_{k\times k}$ and
a $k\times (r=n-k)$ encoding matrix $\mathbf{P}_{k\times r}$, i.e.,
\begin{equation}
\mathbf{G}_{k\times n}=\begin{bmatrix}
\mathbf{I}_{k\times k} & \mathbf{P}_{k\times r} \end{bmatrix},
\label{eq:generator}
\end{equation}
where
\[
\mathbf{P}_{k\times r}=\begin{bmatrix}
p_{1,1} & p_{1,2} & \cdots & p_{1,r}\\
p_{2,1} & p_{2,2} & \cdots & p_{2,r}\\
\vdots & \vdots & \ddots & \vdots \\
p_{k,1} & p_{k,2} & \cdots & p_{k,r}\\
\end{bmatrix}.
\]
To ensure the MDS property, all the square sub-matrices of $\mathbf{P}_{k\times r}$ should be invertible.

\subsection{The Transformation}
\label{sec:transformation}
Next we present a transformation on an $(n,k)$ MDS code to generate an $(n,k)$ transformed MDS  code with $\alpha=(d-k+1)$. In the transformed codes, we have $(d-k+1)\cdot k$ data symbols
$s^\ell_1,s^\ell_2,\ldots,s^\ell_k$ with $\ell=1,2,\ldots,d-k+1$, where $k+1\leq d\leq k+r-1$.
We can compute $(d-k+1)\cdot n$ {\em coded symbols} $c^\ell_1,c^\ell_2,\ldots,c^\ell_n$ by
\[
\begin{bmatrix}
c^\ell_1 & c^\ell_2 & \cdots & c^\ell_{n}
\end{bmatrix}=
\begin{bmatrix}
s^\ell_1 & s^\ell_2 & \cdots & s^\ell_{k}
\end{bmatrix}\cdot \mathbf{G}_{k\times n},
\]
where $\mathbf{G}_{k\times n}$ is given in \eqref{eq:generator} and $\ell=1,2,\ldots,d-k+1$. Let $t=d-k+1$ and denote $\eta$ as
\[
\eta=\left\lfloor \frac{r-1}{d-k}\right\rfloor.
\]
For $i=1,2,\ldots, t$ and $j=1,2,\ldots,\eta$, node $(j-1)t+i$ stores the following $t$
symbols
\begin{equation}
\begin{array}{ll}
& c^1_{(j-1)t+i}+c^i_{(j-1)t+1},\\
& c^2_{(j-1)t+i}+c^i_{(j-1)t+2},\ldots,\\
& c^{i-1}_{(j-1)t+i}+c^i_{(j-1)t+i-1},\\
& c^i_{(j-1)t+i},\\
& c^{i+1}_{(j-1)t+i}+e_jc^i_{(j-1)t+i+1},\\
& c^{i+2}_{(j-1)t+i}+e_jc^i_{(j-1)t+i+2},\ldots,\\
& c^{t}_{(j-1)t+i}+e_jc^i_{(j-1)t+t},
\end{array}
\label{eq:trans-symbol}
\end{equation}
where $e_j$ is an element from the finite field except zero and one.
For $h=t\cdot \eta+1, t\cdot \eta+2,\ldots,n$, node $h$ stores $t$ symbols
\[
c^1_{h},c^2_{h},\ldots,c^{t}_{h}.
\]
The obtained codes are called {\em transformed codes}, which are denoted as $\mathcal{C}_1(n,k,\eta,t)$. Note that the transformation
in \cite{li2017} can be viewed as a special case of our transformation with $d=n-1$ and
$\eta=1$. Table \ref{table:A1} shows an example of the
transformed codes with $n=11$, $k=6$, $r=5$, $d=8$ and $\eta=2$.

\begin{table*}[!t]
\caption{The transformed codes with $n=11$,
$k=6$, $r=5$, $d=8$ and $\eta=2$.}
%\vspace{-12pt}
\scriptsize
\begin{center}
\begin{tabular}{|c|c|c|c|c|c|c|c|c|c|c|}
\hline
Node 1 & Node 2  & Node 3  & Node 4 & Node 5  & Node 6  & Node 7  & Node 8 & Node 9  & Node 10  & Node 11 \\
\hline
$c^1_{1}$& $c^1_{2}+c^2_{1}$ & $c^1_{3}+c^3_{1}$& $c^1_{4}$ & $c^1_{5}+c^2_{4}$ & $c^1_{6}+c^3_{4}$ & $c^1_{7}$ & $c^1_{8}$ & $c^1_{9}$ & $c^1_{10}$ & $c^1_{11}$ \\
\hline
$c^2_{1}+e_1c^1_2$& $c^2_{2}$ &$c^2_{3}+c^3_2$& $c^2_{4}+e_2c^1_{5}$ & $c^2_{5}$ & $c^2_{6}+c^3_{5}$ & $c^2_{7}$ & $c^2_{8}$ & $c^2_{9}$ & $c^2_{10}$ & $c^2_{11}$ \\
\hline
$c^3_{1}+e_1c^1_3$& $c^3_2+e_1c^2_{3}$ &$c^3_{3}$& $c^3_{4}+e_2c^1_{6}$ & $c^3_{5}+e_2c^2_{6}$ & $c^3_{6}$ & $c^3_{7}$ & $c^3_{8}$ & $c^3_{9}$ & $c^3_{10}$ & $c^3_{11}$ \\
\hline
\end{tabular}
\end{center}
\label{table:A1}
\end{table*}

\begin{lemma}
Given integers $\ell$, $i$ and $j$, we can compute
\begin{enumerate}
\item $c^i_\ell$ and $c^j_\ell$ from $c^i_\ell+c^j_\ell$ and $c^i_\ell+ec^j_\ell$, if $e\neq 1$;
\item $c^i_\ell+c^j_\ell$ from $c^i_\ell+ec^j_\ell$ and $c^i_\ell$, if $e\neq 0$;
\item $c^i_\ell+ec^j_\ell$ from $c^i_\ell+c^j_\ell$ and $c^i_\ell$, if $e\neq 0$.
\end{enumerate}
\label{lm:trans}
\end{lemma}
\begin{proof}
Consider the first claim, we can obtain $c^j_\ell$ by
\begin{align*}
c^j_\ell=\{(c^i_\ell+c^j_\ell)-(c^i_\ell+ec^j_\ell)\}{(1-e)^{-1}},
\end{align*}
and further compute $c^i_\ell$ by $(c^i_\ell+c^j_\ell)-c^j_\ell$. The other two claims can be proved similarly.
\end{proof}

\subsection{Optimal Repair Access}

We show in the next theorem that the first $t\cdot \eta$ nodes of the transformed codes have optimal repair access.

\begin{theorem}
We can recover each of the first $t \cdot \eta$ nodes of the transformed codes by accessing $d$ symbols from specified  $d$ nodes.
\label{thm:opt-repair}
\end{theorem}
\begin{proof}
For $i=1,2,\ldots,t$ and $j=1,2,\ldots,\eta$, we show that we can recover $t$ symbols stored in node $(j-1)t+i$ by accessing $k$ symbols $c^{i}_{h_1},c^{i}_{h_2},\ldots,c^{i}_{h_k}$ with
\begin{equation}
\{h_1,\ldots,h_k\}\subset \left\{\begin{array}{ll}
&\{t+i,\ldots,(\eta-1)t+i,t\eta+1,t\eta+2,\ldots,n\} \text{ for } j=1\\
&\{i,\ldots,(j-2)t+i,jt+i,\ldots,(\eta-1)t+i,t\eta+1,\ldots,n\} \text{ for } \eta-1 \geq  j\geq 2\\
&\{i,t+i,\ldots,(\eta-2)t+i,t\eta+1,t\eta+2,\ldots,n\} \text{ for } j=\eta
\end{array}\right.
\label{eq:k-set}
\end{equation}
and $d-k$ symbols
\begin{equation}
\begin{array}{ll}
&c^i_{(j-1)t+1}+e_jc^1_{(j-1)t+i},\ldots,c^{i}_{(j-1)t+i-1}+e_jc^{i-1}_{(j-1)t+i},\\
&c^i_{(j-1)t+i+1}+c^{i+1}_{(j-1)t+i},\ldots,c^i_{(j-1)t+t}+c^{t}_{(j-1)t+i}.
\label{eq:d-k-symbol}
\end{array}
\end{equation}
Note that
\begin{align*}
\eta-1+(n-t\eta)=&n-(d-k)\left\lfloor \frac{n-k-1}{d-k}\right\rfloor -1\\
\geq & k,
\end{align*}
we can thus choose $k$ different values in \eqref{eq:k-set}.

Recall that the $t$ symbols stored in node $(j-1)t+i$ are given in \eqref{eq:trans-symbol}.
By accessing $c^{i}_{h_1},c^{i}_{h_2},\ldots,c^{i}_{h_k}$, we can compute $c^{i}_{(j-1)t+1},c^{i}_{(j-1)t+2},\ldots,c^{i}_{(j-1)t+t}$ according to the MDS property. With the computed $c^{i}_{(j-1)t+1},c^{i}_{(j-1)t+2},\ldots,c^{i}_{(j-1)t+t}$ and the accessed $d-k$ symbols in \eqref{eq:d-k-symbol}, we can compute all $t$ symbols stored in node $(j-1)t+i$ by
Lemma \ref{lm:trans}. Therefore, we can recover node $(j-1)t+i$ by downloading $d$ symbols from $d$ helper nodes and the repair access of node $(j-1)t+i$ is optimal according to \eqref{eq:optimal-repair}.
\end{proof}

Consider the example in Table \ref{table:A1}. We can repair the three symbols $c^1_1$, $c^2_1+e_1c^1_2$ and $c^3_1+e_1c^1_3$ in node 1 by downloading the following eight symbols
\[
c^1_4,c^1_7,c^1_8,c^1_9,c^1_{10},c^1_{11},c^1_2+c^2_1,c^1_3+c^3_1.
\]
Specifically, we can first compute $c^1_1$, $c^1_2$ and $c^1_3$ from the first six symbols of the above downloaded symbols, and then compute $c^2_1+e_1c^1_2$ and $c^3_1+e_1c^1_3$ by
\begin{align*}
c^2_1+e_1c^1_2=&(c^1_2+c^2_1)+(e_1-1)c^1_2,\\
c^3_1+e_1c^1_3=&(c^1_3+c^3_1)+(e_1-1)c^1_3.
\end{align*}
We can repair node 2 and node 3 by downloading
\[
c^2_5,c^2_7,c^2_8,c^2_9,c^2_{10},c^2_{11},c^2_1+e_1c^1_2,c^2_3+c^3_2,
\]
and
\[
c^3_6,c^3_7,c^3_8,c^3_9,c^3_{10},c^3_{11},c^3_1+e_1c^1_3,c^3_2+e_1c^2_3,
\]
respectively. Similarly, we can repair node 4, node 5 and node 6 by downloading
\[
c^1_1,c^1_7,c^1_8,c^1_9,c^1_{10},c^1_{11},c^1_5+c^2_4,c^1_6+c^3_4,
\]
\[
c^2_2,c^2_7,c^2_8,c^2_9,c^2_{10},c^2_{11},c^2_4+e_2c^1_5,c^2_6+c^3_5,
\]
and
\[
c^3_3,c^3_7,c^3_8,c^3_9,c^3_{10},c^3_{11},c^3_4+e_2c^1_6,c^3_5+e_2c^2_6,
\]
respectively.

\subsection{The MDS Property}
\label{sec:mds}
If we put $tk$ data symbols in a vector $\mathbf{s}$ of length $tk$, i.e.,
\[
\mathbf{s}=\begin{bmatrix}
s^1_1 & \cdots & s^1_k & s^2_1 & \cdots  & s^2_k & \cdots  & s^{t}_1 & \cdots  & s^{t}_{k}
\end{bmatrix},
\]
then all $tn$ symbols stored in $n$ nodes are computed as the multiplication of $\mathbf{s}$ and the
$tk\times tn$ generator matrix $\mathbf{G}_{tk \times tn}$.
We can write the matrix $\mathbf{G}_{tk \times tn}$ as
\[
\mathbf{G}_{tk \times tn}=\begin{bmatrix}
\mathbf{G}^1_{tk \times t} & \mathbf{G}^2_{tk \times t} & \cdots & \mathbf{G}^n_{tk \times t}
\end{bmatrix},
\]
where $\mathbf{G}^i_{tk \times t}$ is the generator matrix of node $i$ with $i=1,2,\ldots,n$.
%For $j=1,2,\ldots,\eta$ and $i=1,2,\ldots,t$,
%\[
%\mathbf{G}^{(j-1)t+i}_{tk \times t}=
%\begin{bmatrix}
%\mathbf{0}_{(j-1)t\times t}\\
%\mathbf{P}^i_{t\times t}\\
%\mathbf{0}_{(k-j)t\times t}\\
%\end{bmatrix},
%\]
%where $\mathbf{0}_{i\times j}$ is the $i\times j$ zero matrix and
%\[
%\mathbf{P}^i_{t\times t}=
%\]
In the example in Table \ref{table:A1}, the generator matrix is
\begin{align*}
\mathbf{G}_{18 \times 33}=&\begin{bmatrix}
\mathbf{G}^1_{18 \times 3} & \mathbf{G}^2_{18 \times 3} & \cdots & \mathbf{G}^{11}_{18 \times 3}
\end{bmatrix},
\end{align*}
where
\begin{align*}
&\begin{bmatrix}
\mathbf{G}^1_{18 \times 3} & \mathbf{G}^2_{18 \times 3} & \mathbf{G}^3_{18 \times 3} &\mathbf{G}^4_{18 \times 3} & \mathbf{G}^5_{18 \times 3} &   \mathbf{G}^{6}_{18 \times 3}
\end{bmatrix}\\
=&\left[\begin{array}{ccc|ccc|ccc|ccc|ccc|ccc}
1 & 0 & 0 & 0 & 0 & 0 & 0 & 0 & 0 & 0 & 0 & 0 & 0 & 0 & 0 & 0 & 0 & 0\\
0 & e_1 & 0 & 1 & 0 & 0 & 0 & 0 & 0 & 0 & 0 & 0 & 0 & 0 & 0 & 0 & 0 & 0\\
0 & 0 & e_1 & 0 & 0 & 0 & 1 & 0 & 0 & 0 & 0 & 0 & 0 & 0 & 0 & 0 & 0 & 0\\
0 & 0 & 0 & 0 & 0 & 0 & 0 & 0 & 0 & 1 & 0 & 0 & 0 & 0 & 0 & 0 & 0 & 0\\
0 & 0 & 0 & 0 & 0 & 0 & 0 & 0 & 0 & 0 & e_2 & 0 & 1 & 0 & 0 & 0 & 0 & 0\\
0 & 0 & 0 & 0 & 0 & 0 & 0 & 0 & 0 & 0 & 0 & e_2 & 0 & 0 & 0 & 1 & 0 & 0\\
0 & 1 & 0 & 1 & 0 & 0 & 0 & 0 & 0 & 0 & 0 & 0 & 0 & 0 & 0 & 0 & 0 & 0\\
0 & 0 & 0 & 0 & 1 & 0 & 0 & 0 & 0 & 0 & 0 & 0 & 0 & 0 & 0 & 0 & 0 & 0\\
0 & 0 & 0 & 0 & 0 & e_1 & 0 & 1 & 0 & 0 & 0 & 0 & 0 & 0 & 0 & 0 & 0 & 0\\
0 & 0 & 0 & 0 & 0 & 0 & 0 & 0 & 0 & 0 & 1 & 0 & 1 & 0 & 0 & 0 & 0 & 0\\
0 & 0 & 0 & 0 & 0 & 0 & 0 & 0 & 0 & 0 & 0 & 0 & 0 & 1 & 0 & 0 & 0 & 0\\
0 & 0 & 0 & 0 & 0 & 0 & 0 & 0 & 0 & 0 & 0 & 0 & 0 & 0 & e_2 & 0 & 1 & 0\\
0 & 0 & 1 & 0 & 0 & 0 & 1 & 0 & 0 & 0 & 0 & 0 & 0 & 0 & 0 & 0 & 0 & 0\\
0 & 0 & 0 & 0 & 0 & 1 & 0 & 1 & 0 & 0 & 0 & 0 & 0 & 0 & 0 & 0 & 0 & 0\\
0 & 0 & 0 & 0 & 0 & 0 & 0 & 0 & 1 & 0 & 0 & 0 & 0 & 0 & 0 & 0 & 0 & 0\\
0 & 0 & 0 & 0 & 0 & 0 & 0 & 0 & 0 & 0 & 0 & 1 & 0 & 0 & 0 & 1 & 0 & 0\\
0 & 0 & 0 & 0 & 0 & 0 & 0 & 0 & 0 & 0 & 0 & 0 & 0 & 0 & 1 & 0 & 1 & 0\\
0 & 0 & 0 & 0 & 0 & 0 & 0 & 0 & 0 & 0 & 0 & 0 & 0 & 0 & 0 & 0 & 0 & 1\\
\end{array}\right].
\end{align*}

In the following, we show that we can always find $e_1,e_2,\ldots,e_\eta$ such that the determinant of the $tk\times tk$ matrix composed of any $k$ out of the $n$ generator matrices
\begin{equation}
\mathbf{G}^1_{tk \times t}, \mathbf{G}^2_{tk \times t}, \cdots, \mathbf{G}^n_{tk \times t}
\label{eq:gene-matrix}
\end{equation}
is non-zero, when the field size is large enough.
We first review the Schwartz-Zippel Lemma, and then present the MDS property
condition.
\begin{lemma} (Schwartz-Zippel \cite{motwani1995})
Let $Q(x_1,\ldots,x_n)\in \mathbb{F}_q[x_1,\ldots,x_n]$ be a non-zero multivariate polynomial
of total degree $d$. Let $r_1,\ldots,r_n$ be chosen independently and uniformly at
random from a subset $\mathbb{S}$ of $\mathbb{F}_q$. Then
\begin{equation}
Pr[Q(r_1,\ldots,r_n)=0]\leq \frac{d}{|\mathbb{S}|}.
\end{equation}
\end{lemma}

The MDS property condition is given in the following theorem.

\begin{theorem}
If the field size $q$ is larger than
\begin{equation}
\begin{array}{c}
\eta \frac{(t-1)t}{2}\Bigg(\dbinom{n}{k}-\displaystyle\sum_{\ell=0}^{\eta}\dbinom{n-\eta t}{k-\ell t}\cdot \dbinom{\eta}{\ell}
\Bigg),
\label{field-size2}
\end{array}
\end{equation}
then there exist $\eta$ variables $e_1,e_2,\ldots,e_\eta$ over
$\mathbb{F}_q$ such that the determinant of the $tk\times tk$ matrix composed of any $k$ generator matrices in \eqref{eq:gene-matrix} is non-zero.
\label{thm:hybridcons}
\end{theorem}
\begin{proof}
We view each entry of the matrix $\mathbf{G}_{tk\times tn}$ as a constant and $e_1,e_2,\ldots,e_\eta$ as variables.
We need to evaluate $\dbinom{n}{k}$ determinants that correspond to the determinants of the matrices composed of any $k$ out of $n$ matrices in \eqref{eq:gene-matrix} to be non-zero element in $\mathbb{F}_q$.
%According to Lemma \ref{lm:trans}, we can compute $t^2$ coded symbols from nodes $(j-1)t+1$ to $jt$, where $j=1,2,\ldots,\eta$.
We divide the first $\eta t$ nodes into $\eta$ groups, where group $j$
contains $t$ nodes that are from node $(j-1)t+1$ to node $jt$ with
$j=1,2,\ldots,\eta$. We first prove the following simple lemma.
\begin{lemma}
For $j=1,2,\ldots,\eta$, we can compute $t^2$ coded symbols
\[
c^\ell_{(j-1)t+1},c^\ell_{(j-1)t+2},\cdots,c^\ell_{jt},
\]
for $\ell=1,2,\ldots,t$, from $t^2$ symbols
stored in $t$ nodes of group $j$.
\label{lm-group}
\end{lemma}
\begin{proof}
Let us consider the following $t$ nodes of group $j$:
\begin{align*}
\begin{bmatrix}
\text{Node } (j-1)t+1  & \text{Node } (j-1)t+2  & \cdots  & \text{Node } j t \\
c^1_{(j-1)t+1} &  c^1_{(j-1)t+2}+c^2_{(j-1)t+1} & \cdots & c^1_{j t}+c^t_{(j-1)t+1} \\
c^2_{(j-1)t+1}+e_{j}c^1_{(j-1)t+2} &  c^2_{(j-1)t+2} & \cdots &  c^2_{j t}+c^t_{(j-1)t+2} \\
c^3_{(j-1)t+1}+e_jc^1_{(j-1)t+3} & c^3_{(j-1)t+2}+e_{j}c^2_{(j-1)t+3} & \cdots & c^3_{j t}+c^t_{(j-1)t+3} \\
\vdots &  \vdots &\ddots& \vdots \\
c^t_{(j-1)t+1}+e_{j}c^1_{j t} & c^t_{(j-1)t+2}+e_j c^2_{j t} & \cdots & c^t_{j t} \\
\end{bmatrix}.
\end{align*}
For $i=1,2,\ldots,t-1$ and $\ell=i+1,i+2,\ldots,t$, the symbol in row $\ell$ stored in node $(j-1)t+i$ is $c^\ell_{(j-1)t+i}+e_{j}c^i_{(j-1)t+\ell}$ and the symbol in row $i$ stored in node $(j-1)t+\ell$ is $c^i_{(j-1)t+\ell}+c^\ell_{(j-1)t+i}$. By Lemma \ref{lm:trans}, we can compute two coded symbols $c^i_{(j-1)t+\ell}$ and $c^\ell_{(j-1)t+i}$ from $c^\ell_{(j-1)t+i}+e_{j}c^i_{(j-1)t+\ell}$ and $(j-1)t+\ell$ is $c^i_{(j-1)t+\ell}+c^\ell_{(j-1)t+i}$. Recall that the symbol in row $\ell$ stored in node $(j-1)t+\ell$ is the coded symbol $c^\ell_{(j-1)t+\ell}$. Therefore, we can compute the $t^2$ coded symbols from $t^2$ symbols stored in $t$ nodes of group $j$.
\end{proof}
If the $k$ nodes are composed of $k-\ell t$ nodes from the last $n-\eta t$ nodes and $\ell t$ nodes from $\ell$ out of $\eta$ groups, then we can obtain $tk$ coded symbols by Lemma \ref{lm-group} and further compute $tk$ data symbols, where $\ell\leq \eta$. Therefore, we only need to evaluate
\[
\dbinom{n}{k}-\sum_{\ell=0}^{\eta}\dbinom{n-\eta t}{k-\ell t}\cdot \dbinom{\eta}{\ell}
\]
determinants to be non-zero in $\mathbb{F}_q$ or not. Note that there are $\eta$ variables $e_1,e_2,\ldots,e_\eta$ and each of the $\eta$ variables  appears in $(t-1)t/2$ entries of the generator matrix $\mathbf{G}_{tk\times tn}$. Therefore,
each determinant is a polynomial with  at most degree $\eta (t-1)t/2$. The multiplication of all the determinants can
be interpreted as a polynomial with total degree
\begin{align*}
\eta \frac{(t-1)t}{2}\Bigg(\dbinom{n}{k}-\sum_{\ell=0}^{\eta}\dbinom{n-\eta t}{k-\ell t}\cdot \dbinom{\eta}{\ell}
\Bigg).
\end{align*}
Therefore, the MDS property condition is satisfied if the
field size is larger than~\eqref{field-size2} according to the Schwartz-Zippel Lemma.
\end{proof}

Table \ref{table:A2} shows the underlying field size such that there exists at
least one assignment of $e_1,e_2,\ldots,e_\eta$ with the transformed codes
satisfying the MDS property.

\begin{table*}
%\scriptsize
\caption{The values in \eqref{field-size2} with
some specific parameters.}
%\vspace{-12pt}
\begin{center}
\begin{tabular}{|c|c|c|c|}
\hline
Parameters $(n,k,\eta,t)$ & (8,5,2,2)  & (9,6,2,2)  & (14,10,2,2) \\
\hline
Values in \eqref{field-size2}& 92 & 128& 1400 \\
\hline
\end{tabular}
\end{center}
\label{table:A2}
\end{table*}

According to Theorem \ref{thm:opt-repair}, the transformed codes have optimal repair access for each of the first $\eta t$ nodes.
Similarly, we can also apply the transformation for nodes from $\eta t+1$ to $2\eta t$ such that the repair access of each of nodes from $\eta t+1$ to $2\eta t$ is optimal.

\subsection{Systematic Codes}
Note that the transformed code $\mathcal{C}_1(n,k,\eta,t)$ given in Section \ref{sec:transformation} is not a systematic code. It is important to obtain the systematic transformed code in practice. We can obtain a systematic code by replacing $c^{\ell}_{(j-1)t+i}+c^{i}_{(j-1)t+\ell}$ with $\ell<i$ by $\bar{c}^{\ell}_{(j-1)t+i}$ and replacing $c^{\ell}_{(j-1)t+i}+e_jc^{i}_{(j-1)t+\ell}$ with $\ell>i$ by $\bar{c}^{\ell}_{(j-1)t+i}$, i.e.,
\begin{equation}
\begin{array}{ll}
& \bar{c}^1_{(j-1)t+i}=c^1_{(j-1)t+i}+c^i_{(j-1)t+1},\\
& \bar{c}^2_{(j-1)t+i}=c^2_{(j-1)t+i}+c^i_{(j-1)t+2},\ldots,\\
& \bar{c}^{i-1}_{(j-1)t+i}=c^{i-1}_{(j-1)t+i}+c^i_{(j-1)t+i-1},\\
& \bar{c}^{i+1}_{(j-1)t+i}=c^{i+1}_{(j-1)t+i}+e_jc^i_{(j-1)t+i+1},\\
& \bar{c}^{i+2}_{(j-1)t+i}=c^{i+2}_{(j-1)t+i}+e_jc^i_{(j-1)t+i+2},\ldots,\\
& \bar{c}^{t}_{(j-1)t+i}=c^{t}_{(j-1)t+i}+e_jc^i_{(j-1)t+t},
\end{array}
\label{eq:sys-symbol}
\end{equation}
where $i=1,2,\ldots,t$ and $j=1,2,\ldots,\eta$. For $\ell<i$, we have
\begin{align*}
\bar{c}^{\ell}_{(j-1)t+i}=&c^{\ell}_{(j-1)t+i}+c^{i}_{(j-1)t+\ell},\\
\bar{c}^{i}_{(j-1)t+\ell}=&c^{i}_{(j-1)t+\ell}+e_jc^{\ell}_{(j-1)t+i},
\end{align*}
and further obtain
\begin{align*}
c^{\ell}_{(j-1)t+i}=&\frac{\bar{c}^{i}_{(j-1)t+\ell}-\bar{c}^{\ell}_{(j-1)t+i}}{e_j-1},\\
c^{i}_{(j-1)t+\ell}=&\frac{e_j\bar{c}^{\ell}_{(j-1)t+i}-\bar{c}^{i}_{(j-1)t+\ell}}{e_j-1}.
\end{align*}
In the above equation, $\frac{1}{e_j-1}$ is the inverse of $e_j-1$ over the finite field.
Recall that $c^\ell_h=s^\ell_h$ for $\ell=1,2,\ldots,t$ and $h=1,2,\ldots,k$, and
\[
c^\ell_h=p_{1,h-k}c^\ell_1+p_{2,h-k}c^\ell_2+\ldots+p_{k,h-k}c^\ell_k,
\]
for $h=k+1,k+2,\ldots,n$ and $\ell=1,2,\ldots,t$. We can thus obtain that
\begin{align*}
c^\ell_h=&\Big(\sum_{i=1}^{\ell-1}p_{i,h-k}\frac{e_1\bar{c}^{i}_{\ell}-\bar{c}^{\ell}_{i}}{e_1-1}\Big)+p_{\ell,h-k}c^\ell_\ell+
\Big(\sum_{i=\ell+1}^{t}p_{i,h-k}\frac{\bar{c}^{i}_{\ell}-\bar{c}^{\ell}_{i}}{e_1-1}\Big)+\cdots+\\
&\Big(\sum_{i=1}^{\ell-1}p_{(\eta-1)t+i,h-k}\frac{e_\eta\bar{c}^{i}_{(\eta-1)t+\ell}-\bar{c}^{\ell}_{(\eta-1)t+i}}{e_\eta-1}\Big)+p_{(\eta-1)t+\ell,h-k}c^\ell_{(\eta-1)t+\ell}+\\
&\Big(\sum_{i=\ell+1}^{t}p_{(\eta-1)t+i,h-k}\frac{\bar{c}^{i}_{(\eta-1)t+\ell}-\bar{c}^{\ell}_{(\eta-1)t+i}}{e_\eta-1}\Big)+
p_{\eta t+1,h-k}c^{\ell}_{\eta t+1}+\ldots+p_{k,h-k}c^\ell_k.
\end{align*}

\begin{table*}[!t]
\caption{The systematic transformed codes with $n=11$,
$k=6$, $r=5$, $d=8$ and $\eta=2$.}
%\vspace{-12pt}
\scriptsize
\begin{center}
\begin{tabular}{|c|c|c|c|c|c|c|}
\hline
Node 1 & Node 2  & Node 3  & Node 4 & Node 5  & Node 6  & Node 7  \\
\hline
$c^1_{1}$& $c^1_{2}$ & $c^1_{3}$& $c^1_{4}$ & $c^1_{5}$ & $c^1_{6}$ & $p_{1,1}c^1_1+p_{2,1}\frac{c^2_1-c^1_2}{e_1-1}+p_{3,1}\frac{c^3_1-c^1_3}{e_1-1}+p_{4,1}c^1_4+p_{5,1}\frac{c^2_4-c^1_5}{e_2-1}+p_{6,1}\frac{c^3_4-c^1_6}{e_2-1}$  \\
\hline
$c^2_{1}$& $c^2_{2}$ &$c^2_{3}$& $c^2_{4}$ & $c^2_{5}$ & $c^2_{6}$ & $p_{1,1}\frac{e_1c^1_2-c^2_1}{e_1-1}+p_{2,1}c^2_2+p_{3,1}\frac{c^3_2-c^2_3}{e_1-1}+p_{4,1}\frac{e_2c^1_5-c^2_4}{e_2-1}+p_{5,1}c^2_5+p_{6,1}\frac{c^3_5-c^2_6}{e_2-1}$  \\
\hline
$c^3_{1}$& $c^3_2$ &$c^3_{3}$& $c^3_{4}$ & $c^3_{5}$ & $c^3_{6}$ & $p_{1,1}\frac{e_1c^1_3-c^3_1}{e_1-1}+p_{2,1}\frac{e_1c^2_3-c^3_2}{e_1-1}+p_{3,1}c^3_3+p_{4,1}\frac{e_2c^1_6-c^3_4}{e_2-1}+p_{5,1}\frac{e_2c^2_6-c^3_5}{e_2-1}+p_{6,1}c^3_6$  \\
\hline
\end{tabular}
\begin{tabular}{|c|c|}
\hline
Node 8 & Node 9   \\
\hline
 $p_{1,2}c^1_1+\sum_{i=2}^{3}p_{i,2}\frac{c^{i}_1-c^1_{i}}{e_1-1}+p_{4,2}c^1_4+\sum_{i=2}^{3}p_{3+i,2}\frac{c^i_4-c^1_{3+i}}{e_2-1}$ & $p_{1,3}c^1_1+\sum_{i=2}^{3}p_{i,3}\frac{c^{i}_1-c^1_{i}}{e_1-1}+p_{4,3}c^1_4+\sum_{i=2}^{3}p_{3+i,3}\frac{c^i_4-c^1_{3+i}}{e_2-1}$  \\
\hline
$p_{1,2}\frac{e_1c^1_2-c^2_1}{e_1-1}+p_{2,2}c^2_2+p_{4,2}\frac{e_2c^1_6-c^3_4}{e_2-1}+\sum_{i=1}^{2}p_{3i,2}\frac{c^{3}_{3i-1}-c^2_{3i}}{e_i-1}+p_{5,2}c^2_5$ &$p_{1,3}\frac{e_1c^1_2-c^2_1}{e_1-1}+p_{2,3}c^2_2+p_{4,3}\frac{e_2c^1_6-c^3_4}{e_2-1}+\sum_{i=1}^{2}p_{3i,3}\frac{c^{3}_{3i-1}-c^2_{3i}}{e_i-1}+p_{5,3}c^2_5$  \\
\hline
$\sum_{i=1}^{2}p_{i,2}\frac{e_1c^i_3-c^3_i}{e_1-1}+p_{3,2}c^3_3+\sum_{i=1}^{2}p_{3+i,2}\frac{e_2c^i_6-c^3_{3+i}}{e_2-1}+p_{6,2}c^3_6$ & $\sum_{i=1}^{2}p_{i,3}\frac{e_1c^i_3-c^3_i}{e_1-1}+p_{3,3}c^3_3+\sum_{i=1}^{2}p_{3+i,3}\frac{e_2c^i_6-c^3_{3+i}}{e_2-1}+p_{6,3}c^3_6$ \\
\hline
\end{tabular}
\begin{tabular}{|c|c|}
\hline
Node 10 & Node 11   \\
\hline
 $p_{1,4}c^1_1+\sum_{i=2}^{3}p_{i,4}\frac{c^{i}_1-c^1_{i}}{e_1-1}+p_{4,4}c^1_4+\sum_{i=2}^{3}p_{3+i,4}\frac{c^i_4-c^1_{3+i}}{e_2-1}$ & $p_{1,5}c^1_1+\sum_{i=2}^{3}p_{i,5}\frac{c^{i}_1-c^1_{i}}{e_1-1}+p_{4,5}c^1_4+\sum_{i=2}^{3}p_{3+i,5}\frac{c^i_4-c^1_{3+i}}{e_2-1}$  \\
\hline
$p_{1,4}\frac{e_1c^1_2-c^2_1}{e_1-1}+p_{2,4}c^2_2+p_{4,4}\frac{e_2c^1_6-c^3_4}{e_2-1}+\sum_{i=1}^{2}p_{3i,4}\frac{c^{3}_{3i-1}-c^2_{3i}}{e_i-1}+p_{5,4}c^2_5$ &$p_{1,5}\frac{e_1c^1_2-c^2_1}{e_1-1}+p_{2,5}c^2_2+p_{4,5}\frac{e_2c^1_6-c^3_4}{e_2-1}+\sum_{i=1}^{2}p_{3i,5}\frac{c^{3}_{3i-1}-c^2_{3i}}{e_i-1}+p_{5,5}c^2_5$  \\
\hline
$\sum_{i=1}^{2}p_{i,4}\frac{e_1c^i_3-c^3_i}{e_1-1}+p_{3,4}c^3_3+\sum_{i=1}^{2}p_{3+i,4}\frac{e_2c^i_6-c^3_{3+i}}{e_2-1}+p_{6,4}c^3_6$ & $\sum_{i=1}^{2}p_{i,5}\frac{e_1c^i_3-c^3_i}{e_1-1}+p_{3,5}c^3_3+\sum_{i=1}^{2}p_{3+i,5}\frac{e_2c^i_6-c^3_{3+i}}{e_2-1}+p_{6,5}c^3_6$ \\
\hline
\end{tabular}
\end{center}
\label{table:sys-trans}
\end{table*}

%\sum_{i=1}^{2}p_{3i-2,2}\frac{e_ic^1_{3i-1}-c^2_{3i-2}}{e_i-1}
Table \ref{table:sys-trans} shows the systematic transformed code with $n=11$,
$k=6$, $r=5$, $d=8$ and $\eta=2$.

The systematic transformed codes also satisfy Theorem \ref{thm:opt-repair} and Theorem \ref{thm:hybridcons}, as the two transformed codes are equivalent. According to Theorem \ref{thm:opt-repair}, each of the first six nodes in Table \ref{table:sys-trans} has optimal repair access. For example, we can recover the three symbols in node 1 by accessing the following eight symbols
\[
c^1_2,c^1_3,c^1_4,c^1_7,c^1_8,c^1_9,c^1_{10},c^1_{11}.
\]
Specifically, we can first subtract $c^1_4$ from $c^1_7$, $c^1_8$, $c^1_9$, $c^1_{10}$ and $c^1_{11}$, respectively, to obtain
\begin{align*}
&\begin{bmatrix}
p_{1,1}c^1_1+p_{2,1}\frac{c^2_1-c^1_2}{e_1-1}+p_{3,1}\frac{c^3_1-c^1_3}{e_1-1}+p_{5,1}\frac{c^2_4-c^1_5}{e_2-1}+p_{6,1}\frac{c^3_4-c^1_6}{e_2-1}\\
p_{1,2}c^1_1+p_{2,2}\frac{c^2_1-c^1_2}{e_1-1}+p_{3,2}\frac{c^3_1-c^1_3}{e_1-1}+p_{5,2}\frac{c^2_4-c^1_5}{e_2-1}+p_{6,2}\frac{c^3_4-c^1_6}{e_2-1}\\
p_{1,3}c^1_1+p_{2,3}\frac{c^2_1-c^1_2}{e_1-1}+p_{3,3}\frac{c^3_1-c^1_3}{e_1-1}+p_{5,3}\frac{c^2_4-c^1_5}{e_2-1}+p_{6,3}\frac{c^3_4-c^1_6}{e_2-1}\\
p_{1,4}c^1_1+p_{2,4}\frac{c^2_1-c^1_2}{e_1-1}+p_{3,4}\frac{c^3_1-c^1_3}{e_1-1}+p_{5,4}\frac{c^2_4-c^1_5}{e_2-1}+p_{6,4}\frac{c^3_4-c^1_6}{e_2-1}\\
p_{1,5}c^1_1+p_{2,5}\frac{c^2_1-c^1_2}{e_1-1}+p_{3,5}\frac{c^3_1-c^1_3}{e_1-1}+p_{5,5}\frac{c^2_4-c^1_5}{e_2-1}+p_{6,5}\frac{c^3_4-c^1_6}{e_2-1}\\
\end{bmatrix}\\
=&\begin{bmatrix}
c^1_1 &\frac{c^2_1-c^1_2}{e_1-1} &\frac{c^3_1-c^1_3}{e_1-1} &\frac{c^2_4-c^1_5}{e_2-1} &\frac{c^3_4-c^1_6}{e_2-1}
\end{bmatrix}\cdot
\begin{bmatrix}
p_{1,1} &p_{1,2} &p_{1,3} &p_{1,4} &p_{1,5}\\
p_{2,1} &p_{2,2} &p_{2,3} &p_{2,4} &p_{2,5}\\
p_{3,1} &p_{3,2} &p_{3,3} &p_{3,4} &p_{3,5}\\
p_{5,1} &p_{5,2} &p_{5,3} &p_{5,4} &p_{5,5}\\
p_{6,1} &p_{6,2} &p_{6,3} &p_{6,4} &p_{6,5}\\
\end{bmatrix}.
\end{align*}
As the square matrix in the above equation is invertible, we can compute the five symbols
\[
c^1_1,\frac{c^2_1-c^1_2}{e_1-1},\frac{c^3_1-c^1_3}{e_1-1},\frac{c^2_4-c^1_5}{e_2-1},\frac{c^3_4-c^1_6}{e_2-1},
\]
from the above five symbols, and further obtain
\[
c^1_1,c^2_1-c^1_2,c^3_1-c^1_3,c^2_4-c^1_5,c^3_4-c^1_6,
\]
as $e_1-1$ and $e_2-1$ are non-zero in the finite field. We have recovered $c^1_1$. Together with $c^1_2$ and $c^1_3$, we can recover $c^2_1$ and $c^3_1$ by
\begin{align*}
c^2_1=&(c^2_1-c^1_2)+c^1_2,\\
c^3_1=&(c^3_1-c^1_3)+c^1_3.
\end{align*}
Similarly, we can recover the symbols stored in node 2, node 3, node 4, node 5 and node 6 by accessing
\[
c^2_1,c^2_3,c^2_5,c^2_7,c^2_8,c^2_9,c^2_{10},c^2_{11},
\]
\[
c^3_1,c^3_2,c^3_6,c^3_7,c^3_8,c^3_9,c^3_{10},c^3_{11},
\]
\[
c^1_1,c^1_5,c^1_6,c^1_7,c^1_8,c^1_9,c^1_{10},c^1_{11},
\]
\[
c^2_1,c^2_4,c^2_6,c^2_7,c^2_8,c^2_9,c^2_{10},c^2_{11},
\]
and
\[
c^3_3,c^3_4,c^3_5,c^3_7,c^3_8,c^3_9,c^3_{10},c^3_{11}.
\]
respectively.

\section{Multi-Layer Transformed MDS Codes with Optimal Repair Access}
\label{sec:const}

In this section, we present the construction of multi-layer transformed MDS codes by recursively applying the transformation given in Section \ref{sec:trans}.

\subsection{Construction}

We divide $n$ nodes into $\lceil \frac{n}{\eta t}\rceil$ sets, each of which
contains $\eta t$ nodes.  For $i=1,2,\ldots,\lceil \frac{n}{\eta t}\rceil-1$,
set $i$ contains nodes between $(i-1)\eta t+1$ and $i\eta t$, while set
$\lceil \frac{n}{\eta t}\rceil$ contains the last $\eta t$ nodes. We further
divide each set into $\eta$ groups, each of which contains $t$ nodes.

If we apply the transformation in Section \ref{sec:trans} for the first set of an $(n,k)$ MDS code, we can obtain a transformed code $\mathcal{C}_1(n,k,\eta,t)$ with each node having $t$ symbols such that, according to Theorem \ref{thm:hybridcons} and Theorem \ref{thm:opt-repair}, $\mathcal{C}_1(n,k,\eta,t)$ is an MDS code and has optimal repair access for each of the first $\eta t$ nodes. If we
apply the transformation in Section \ref{sec:trans} for the second set of the code $\mathcal{C}_1(n,k,\eta,t)$, we can obtain the code $\mathcal{C}_2(n,k,\eta,t)$ with each node having $\alpha=t^2$ symbols.
Specifically, we can obtain $\mathcal{C}_2(n,k,\eta,t)$ as follows. We first generate $t$ instances of the code $\mathcal{C}_1(n,k,\eta,t)$ and view the $t$ symbols stored in each node of $\mathcal{C}_1(n,k,\eta,t)$ as a vector. For $\ell=1,2,\ldots,t$ and $h=1,2,\ldots,n$, the vector stored in node $h$ of instance $\ell$ of $\mathcal{C}_1(n,k,\eta,t)$ is denoted as $\mathbf{v}^\ell_h$. For $i=1,2,\ldots,t$ and $j=1,2,\ldots,\eta$, node $t\eta+(j-1)t+i$ of $\mathcal{C}_2(n,k,\eta,t)$ stores the following $t$ vectors ($t^2$ symbols)
\begin{equation}
\begin{array}{ll}
& \mathbf{v}^1_{t\eta+(j-1)t+i}+\mathbf{v}^i_{t\eta+(j-1)t+1},\\
& \mathbf{v}^2_{t\eta+(j-1)t+i}+\mathbf{v}^i_{t\eta+(j-1)t+2},\ldots,\\
& \mathbf{v}^{i-1}_{t\eta+(j-1)t+i}+\mathbf{v}^i_{t\eta+(j-1)t+i-1},\\
& \mathbf{v}^i_{t\eta+(j-1)t+i},\\
& \mathbf{v}^{i+1}_{t\eta+(j-1)t+i}+e_{\eta+j}\mathbf{v}^i_{t\eta+(j-1)t+i+1},\\
& \mathbf{v}^{i+2}_{t\eta+(j-1)t+i}+e_{\eta+j}\mathbf{v}^i_{t\eta+(j-1)t+i+2},\ldots,\\
& \mathbf{v}^{t}_{t\eta+(j-1)t+i}+e_{\eta+j}\mathbf{v}^i_{t\eta+(j-1)t+t},
\end{array}
\label{eq:trans-vectors}
\end{equation}
where $e_{\eta+j}$ is an element from the finite field except zero and one. Note that the multiplication of $e_{\eta+j}$ and a vector
\[
\mathbf{v}=\begin{bmatrix}
v_1 & v_2 &\ldots &v_t
\end{bmatrix}
\]
is defined as
\[
e_{\eta+j}\mathbf{v}=\begin{bmatrix}
e_{\eta+j}v_1 & e_{\eta+j}v_2 &\ldots &e_{\eta+j}v_t
\end{bmatrix}
\]
and the addition of two vectors
\[
\mathbf{v}^1=\begin{bmatrix}
v^1_1 & v^1_2 &\ldots &v^1_t
\end{bmatrix}
\]
and
\[
\mathbf{v}^2=\begin{bmatrix}
v^2_1 & v^2_2 &\ldots &v^2_t
\end{bmatrix}
\]
is
\[
\mathbf{v}^1+\mathbf{v}^2=\begin{bmatrix}
v^1_1+v^2_1 & v^1_2+v^2_2 &\ldots &v^1_t+v^2_t
\end{bmatrix}.
\]
For $h\in \{1,2,\ldots,n\}\setminus \{t\cdot \eta+1, t\cdot \eta+2,\ldots,2t\cdot \eta\}$, node $h$ stores $t$ vectors ($t^2$ symbols)
\[
\mathbf{v}^1_{h},\mathbf{v}^2_{h},\ldots,\mathbf{v}^{t}_{h}.
\]
According to Theorem \ref{thm:hybridcons},  $\mathcal{C}_2(n,k,\eta,t)$ is an MDS code and, according to Theorem \ref{thm:opt-repair}, has optimal repair access for each of the second $\eta t$ nodes.

Consider the example of the code $\mathcal{C}_1(n=8,k=5,\eta=2,t=2)$ in Section \ref{sec:example}. We can obtain $\mathcal{C}_2(n=8,k=5,\eta=2,t=2)$ by applying the transformation in Section \ref{sec:trans} for the last four nodes (the second set) of the code $\mathcal{C}_1(n=8,k=5,\eta=2,t=2)$. We should first generate $t=2$ instances of $\mathcal{C}_1(n=8,k=5,\eta=2,t=2)$ that are given by \eqref{eq:two-ins} and the vector ($t=2$ symbols) stored in node $h$ of instance $\ell$ of $\mathcal{C}_1(n=8,k=5,\eta=2,t=2)$ is denoted as $\mathbf{v}^\ell_h$, where $\ell=1,2$ and $h=1,2,\ldots,8$.
We thus have
\begin{align*}
&\mathbf{v}^{\ell}_1=\begin{bmatrix}
c^{2(\ell-1)+1}_1, & c^{2(\ell-1)+2}_1+e_1c^{2(\ell-1)+1}_2
\end{bmatrix},\\
&\mathbf{v}^{\ell}_2=\begin{bmatrix}
c^{2(\ell-1)+1}_2+c^{2(\ell-1)+2}_1, & c^{2(\ell-1)+2}_2
\end{bmatrix},\\
&\mathbf{v}^{\ell}_3=\begin{bmatrix}
c^{2(\ell-1)+1}_3, & c^{2(\ell-1)+2}_3+e_2c^{2(\ell-1)+1}_4
\end{bmatrix},\\
&\mathbf{v}^{\ell}_4=\begin{bmatrix}
c^{2(\ell-1)+1}_4+c^{2(\ell-1)+2}_3, & c^{2(\ell-1)+2}_4
\end{bmatrix},\\
&\mathbf{v}^{\ell}_5=\begin{bmatrix}
c^{2(\ell-1)+1}_5, & c^{2(\ell-1)+2}_5
\end{bmatrix},\\
&\mathbf{v}^{\ell}_6=\begin{bmatrix}
c^{2(\ell-1)+1}_6, & c^{2(\ell-1)+2}_6
\end{bmatrix},\\
&\mathbf{v}^{\ell}_7=\begin{bmatrix}
c^{2(\ell-1)+1}_7, & c^{2(\ell-1)+2}_7
\end{bmatrix},\\
&\mathbf{v}^{\ell}_8=\begin{bmatrix}
c^{2(\ell-1)+1}_8, & c^{2(\ell-1)+2}_8
\end{bmatrix},
\end{align*}
where $\ell=1,2$. According to \eqref{eq:trans-vectors}, the storage of the last four nodes is
\begin{align*}
\left[\begin{array}{c|cc}
\text{Node } 5 &\mathbf{v}^1_5=[c^1_5, c^2_5] & \mathbf{v}^2_5+e_3\mathbf{v}^1_6=[c^3_5+e_3c^1_6, c^4_5+e_3c^2_6]\\ \hline
\text{Node } 6 &\mathbf{v}^1_6+\mathbf{v}^2_5=[c^1_6+c^3_5, c^2_6+c^4_5] & \mathbf{v}^2_6=[c^3_6, c^4_6] \\ \hline
\text{Node } 7 &\mathbf{v}^1_7=[c^1_7, c^2_7] & \mathbf{v}^2_7+e_4\mathbf{v}^1_8=[c^3_7+e_4c^1_8, c^4_7+e_4c^2_8] \\\hline
\text{Node } 8 & \mathbf{v}^1_8+\mathbf{v}^2_7=[c^1_8+c^3_7, c^2_8+c^4_7]& \mathbf{v}^2_8=[c^3_8, c^4_8] \\\hline
\end{array}\right],
\end{align*}
which is the same as the storage of the last four nodes in Table \ref{table:A3}.

\begin{table*}[!t]
%\scriptsize
\caption{The storage of the first $\eta t$ nodes of $\mathcal{C}_1(n,k,\eta,t)$.}
%\vspace{-12pt}
\begin{center}
\begin{tabular}{|c|c|c|c|c|c|c|c|c|}
\hline
Node 1 & Node 2  & $\cdots$  & Node $t$ & $\cdots$  & Node $(\eta-1)t+1$  & Node $(\eta-1)t+2$  & $\cdots$  & Node $\eta t$ \\
\hline
$c^1_{1}$& $c^1_{2}+c^2_{1}$ & $\cdots$& $c^1_{t}+c^t_{1}$ & $\cdots$ & $c^1_{(\eta-1)t+1}$& $c^1_{(\eta-1)t+2}+c^2_{(\eta-1)t+1}$ & $\cdots$& $c^1_{\eta t}+c^t_{(\eta-1)t+1}$ \\
\hline
$c^2_{1}+e_1c^1_2$& $c^2_{2}$ &$\cdots$& $c^2_{t}+c^t_2$ & $\cdots$ & $c^2_{(\eta-1)t+1}+e_{\eta}c^1_{(\eta-1)t+2}$& $c^2_{(\eta-1)t+2}$ &$\cdots$& $c^2_{\eta t}+c^t_{(\eta-1)t+2}$ \\
\hline
$c^3_{1}+e_1c^1_3$& $c^3_{2}+e_1c^2_{3}$ & $\cdots$& $c^3_{t}+c^t_{3}$ & $\cdots$ & $c^3_{(\eta-1)t+1}+e_\eta c^1_{(\eta-1)t+3}$& $c^3_{(\eta-1)t+2}+e_{\eta}c^2_{(\eta-1)t+3}$ & $\cdots$& $c^3_{\eta t}+c^t_{(\eta-1)t+3}$ \\
\hline
$\vdots$& $\vdots$ &$\ddots$& $\vdots$ & $\cdots$ & $\vdots$& $\vdots$ &$\ddots$& $\vdots$ \\
\hline
$c^t_{1}+e_1c^1_t$& $c^t_{2}+e_1c^2_{t}$ & $\cdots$& $c^t_{t}$ & $\cdots$ & $c^t_{(\eta-1)t+1}+e_{\eta}c^1_{\eta t}$& $c^t_{(\eta-1)t+2}+e_\eta c^2_{\eta t}$ & $\cdots$& $c^t_{\eta t}$ \\
\hline
\end{tabular}
\end{center}
\label{table:thm-proof1}
\end{table*}

Table \ref{table:thm-proof1} shows the $t$ symbols stored in each of the first $\eta t$ nodes of $\mathcal{C}_1(n,k,\eta,t)$. For $i=\eta t+1,\ldots,n$, the $t$ symbols stored in node $i$ of $\mathcal{C}_1(n,k,\eta,t)$ are $c^1_i,c^2_i,\ldots,c^t_i$.
%In $\mathcal{C}_2(n,k,\eta,t)$, the $t^2$ symbols stored in node $i$ can be viewed as $t$ copies of the $t$ symbols stored in node $i$ of $\mathcal{C}_1(n,k,\eta,t)$.
Table \ref{table:thm-proof2} shows the storage of the first $2\eta t$ nodes of $\mathcal{C}_2(n,k,\eta,t)$.
In the next theorem, we show that the repair access of each of the first $\eta t$ nodes is also optimal.

\begin{table*}[!t]
\scriptsize
\caption{The storage of the first $2\eta t$ nodes of $\mathcal{C}_2(n,k,\eta,t)$.}
%\vspace{-12pt}
\begin{center}
\begin{tabular}{|c|c|c|c|c|c|c|}
\hline
Node 1 & $\cdots$  & Node $t$ & $\cdots$  & Node $(\eta-1)t+1$   & $\cdots$  & Node $\eta t$ \\
\hline
$c^1_{1}$& $\cdots$& $c^1_{t}+c^t_{1}$ & $\cdots$ & $c^1_{(\eta-1)t+1}$&  $\cdots$& $c^1_{\eta t}+c^t_{(\eta-1)t+1}$ \\
$c^2_{1}+e_1c^1_2$& $\cdots$& $c^2_{t}+c^t_2$ & $\cdots$ & $c^2_{(\eta-1)t+1}+e_{\eta}c^1_{(\eta-1)t+2}$& $\cdots$& $c^2_{\eta t}+c^t_{(\eta-1)t+2}$ \\
$\vdots$& $\ddots$& $\vdots$ & $\cdots$ & $\vdots$ &$\ddots$& $\vdots$ \\
$c^t_{1}+e_1c^1_t$& $\cdots$& $c^t_{t}$ & $\cdots$ & $c^t_{(\eta-1)t+1}+e_{\eta}c^1_{\eta t}$& $\cdots$& $c^t_{\eta t}$ \\
\hline
$\vdots$& $\vdots$& $\vdots$ & $\vdots$ & $\vdots$& $\vdots$ &$\vdots$ \\
\hline
$c^{(t-1)t+1}_{1}$& $\cdots$& $c^{(t-1)t+1}_{t}+c^{(t-1)t+t}_{1}$ & $\cdots$ & $c^{(t-1)t+1}_{(\eta-1)t+1}$& $\cdots$& $c^{(t-1)t+1}_{\eta t}+c^{t^2}_{(\eta-1)t+1}$ \\
$c^{(t-1)t+2}_{1}+e_1c^{(t-1)t+1}_2$& $\cdots$& $c^{(t-1)t+2}_{t}+c^{(t-1)t+t}_2$ & $\cdots$ & $c^{(t-1)t+2}_{(\eta-1)t+1}+e_{\eta}c^{(t-1)t+1}_{(\eta-1)t+2}$& $\cdots$& $c^{(t-1)t+2}_{\eta t}+c^{t^2}_{(\eta-1)t+2}$ \\
$\vdots$& $\ddots$& $\vdots$ & $\cdots$ & $\vdots$ &$\ddots$& $\vdots$ \\
$c^{(t-1)t+t}_{1}+e_1c^{(t-1)t+1}_t$& $\cdots$& $c^{(t-1)t+t}_{t}$ & $\cdots$ & $c^{(t-1)t+t}_{(\eta-1)t+1}+e_{\eta}c^{(t-1)t+1}_{\eta t}$ & $\cdots$& $c^{t^2}_{\eta t}$ \\
\hline
\end{tabular}
\begin{tabular}{|c|c|c|c|c|c|c|}
\hline
Node $\eta t+it+1$ & Node $\eta t+it+2$ & $\cdots$ & Node $\eta t+it +t$\\
\hline
 $c^{1}_{\eta t+it+1}$ &  $c^{1}_{\eta t+it+2}+c^{t+1}_{\eta t+it+1}$ & $\cdots$ &  $c^{1}_{\eta t+it+t}+c^{(t-1)t+1}_{\eta t+it+1}$\\
 $c^{2}_{\eta t+it+1}$ &   $c^{2}_{\eta t+it+2}+c^{t+2}_{\eta t+it+1}$  & $\cdots$ &  $c^{2}_{\eta t+it+t}+c^{(t-1)t+2}_{\eta t+it+1}$\\
 $\vdots$ & $\vdots$  & $\cdots$&$\vdots$ \\
 $c^{t}_{\eta t+it+1}$ &   $c^{t}_{\eta t+it+2}+c^{2t}_{\eta t+it+1}$ & $\cdots$  &  $c^{t}_{\eta t+it+t}+c^{t^2}_{\eta t+it+1}$\\
\hline
  $c^{t+1}_{\eta t+it+1}+e_{\eta+i+1}c^{1}_{\eta t+it+2}$ & $c^{t+1}_{\eta t+it+2}$ & $\cdots$ & $c^{t+1}_{\eta t+it+t}+c^{(t-1)t+1}_{\eta t+it+2}$\\
  $c^{t+2}_{\eta t+it+1}+e_{\eta+i+1}c^{2}_{\eta t+it+2}$ & $c^{t+2}_{\eta t+it+2}$  & $\cdots$& $c^{t+2}_{\eta t+it+t}+c^{(t-1)t+2}_{\eta t+it+2}$\\
  $\vdots$ & $\vdots$ & $\cdots$ & $\vdots$\\
  $c^{2t}_{\eta t+it+1}+e_{\eta+i+1}c^{t}_{\eta t+it+2}$ & $c^{2t}_{\eta t+it+2}$  & $\cdots$& $c^{2t}_{\eta t+it+t}+c^{t^2}_{\eta t+it+2}$\\
\hline
  $\vdots$ & $\vdots$ & $\cdots$ & $\vdots$\\
\hline
  $c^{(t-1)t+1}_{\eta t+it+1}+e_{\eta+i+1}c^{1}_{\eta t+it+t}$&  $c^{(t-1)t+1}_{\eta t+it+2}+e_{\eta+i+1}c^{t+1}_{\eta t+it+t}$ & $\cdots$& $c^{(t-1)t+1}_{\eta t+it+t}$\\
  $c^{(t-1)t+2}_{\eta t+it+1}+e_{\eta+i+1}c^{2}_{\eta t+it+t}$&  $c^{(t-1)t+2}_{\eta t+it+2}+e_{\eta+i+1}c^{t+2}_{\eta t+it+t}$ & $\cdots$& $c^{(t-1)t+2}_{\eta t+it+t}$\\
  $\vdots$& $\vdots$ & $\cdots$& $\vdots$\\
  $c^{t^2}_{\eta t+it+1}+e_{\eta+i+1}c^{t}_{\eta t+it+t}$&  $c^{t^2}_{\eta t+it+2}+e_{\eta+i+1}c^{2t}_{\eta t+it+t}$ & $\cdots$& $c^{t^2}_{\eta t+it+t}$\\
\hline
\end{tabular}
\end{center}
\label{table:thm-proof2}
\end{table*}

\begin{theorem}
The repair access of node $i$ of $\mathcal{C}_2(n,k,\eta,t)$ for $i=1,2,\ldots, \eta t$ is optimal.
\label{thm:repair-preserve}
\end{theorem}
\begin{proof}
For $i=1,2,\ldots,t$ and $j=1,2,\ldots,\eta$, we can repair $t^2$ symbols in node $(j-1)t+i$ by downloading
$\eta t^2$ symbols from nodes $\eta t+1,\eta t+2,\ldots,2\eta t$ in rows $i,i+t,\ldots,i+(t-1)t$, and $(k-\eta t)t$ symbols from
nodes $h_1,\ldots,h_{k-\eta t}$ in rows $i,i+t,\ldots,i+(t-1)t$ with indices $h_1,\ldots,h_k$ in
\begin{equation}
\{h_1,\ldots,h_{k-\eta t}\}\subset
\{i,\ldots,(j-2)t+i,jt+i,\ldots,(\eta-1)t+i,2t\eta+1,\ldots,n\}
\label{eq:k-set1}
\end{equation}
and $(d-k)t$ symbols in \eqref{eq:t(d-k)-symbol}.
\begin{equation}
\begin{array}{ll}
&c^i_{(j-1)t+1}+e_jc^1_{(j-1)t+i},\ldots,c^{i}_{(j-1)t+i-1}+e_jc^{i-1}_{(j-1)t+i},\\
&c^i_{(j-1)t+i+1}+c^{i+1}_{(j-1)t+i},\ldots,c^i_{(j-1)t+t}+c^{t}_{(j-1)t+i},\\
&c^{t+i}_{(j-1)t+1}+e_jc^{t+1}_{(j-1)t+i},\ldots,c^{t+i}_{(j-1)t+i-1}+e_jc^{t+i-1}_{(j-1)t+i},\\
&c^{t+i}_{(j-1)t+i+1}+c^{t+i+1}_{(j-1)t+i},\ldots,c^{t+i}_{(j-1)t+t}+c^{2t}_{(j-1)t+i},\\
& \vdots \\
&c^{(t-1)t+i}_{(j-1)t+1}+e_{j}c^{(t-1)t+1}_{(j-1)t+i},\ldots,c^{(t-1)t+i}_{(j-1)t+i-1}+e_{j}c^{(t-1)t+i-1}_{(j-1)t+i},\\
&c^{(t-1)t+i}_{(j-1)t+i+1}+c^{(t-1)t+i+1}_{(j-1)t+i},\ldots,c^{(t-1)t+i}_{(j-1)t+t}+c^{(t^2}_{(j-1)t+i}.
\label{eq:t(d-k)-symbol}
\end{array}
\end{equation}
Note that
\begin{align*}
\eta-1+(n-2t\eta)=&n-(2t-1)\lfloor \frac{n-k-1}{d-k}\rfloor -1\\
\geq & k-\eta t,
\end{align*}
we can thus choose $k-\eta t$ different values in \eqref{eq:k-set1}.

First, we claim that we can obtain the following $kt$ symbols
\begin{equation}
\begin{array}{ll}
&c^{i}_{\eta t+1},c^{i}_{\eta t+2},\ldots,c^{i}_{2\eta t},c^{i}_{h_1},c^{i}_{h_2},\ldots,c^{i}_{h_{k-\eta t}},\\
&c^{i+t}_{\eta t+1},c^{i+t}_{\eta t+2},\ldots,c^{i+t}_{2\eta t},c^{i+t}_{h_1},c^{i+t}_{h_2},\ldots,c^{i+t}_{h_{k-\eta t}},\\
& \vdots \\
&c^{i+(t-1)t}_{\eta t+1},c^{i+(t-1)t}_{\eta t+2},\ldots,c^{i+(t-1)t}_{2\eta t},c^{i+(t-1)t}_{h_1},c^{i+(t-1)t}_{h_2},\ldots,c^{i+(t-1)t}_{h_{k-\eta t}},
\label{eq:thm-pr}
\end{array}
\end{equation}
from the downloaded $kt$ symbols from nodes $\eta t+1,\eta t+2,\ldots,2\eta t$ and $h_1,\ldots,h_{k-\eta t}$.
As $c^{i}_{h_j},c^{i+t}_{h_j},\ldots,c^{i+(t-1)t}_{h_{j}}$ are directly downloaded from node $h_j$ for $j=1,2,\ldots,k-\eta t$, we only need to show that we can obtain
\begin{equation}
\begin{array}{ll}
&c^{i}_{\eta t+1},c^{i}_{\eta t+2},\ldots,c^{i}_{2\eta t},\\
&c^{i+t}_{\eta t+1},c^{i+t}_{\eta t+2},\ldots,c^{i+t}_{2\eta t},\\
& \vdots \\
&c^{i+(t-1)t}_{\eta t+1},c^{i+(t-1)t}_{\eta t+2},\ldots,c^{i+(t-1)t}_{2\eta t},
\end{array}
\label{eq:thm-pr1}
\end{equation}
from the downloaded $\eta t^2$ symbols from nodes $\eta t+1$ to $2\eta t$.
Recall that the $\eta t^2$ symbols downloaded from nodes $\eta t+1$ to $2\eta t$ are
\begin{align*}
&c^{i}_{\eta t+\ell t+1},c^{i}_{\eta t+\ell t+2}+c^{t+i}_{\eta t+\ell t+1},\ldots,c^{i}_{\eta t+\ell t+t}+c^{(t-1)t+i}_{\eta t+\ell t+1},\\
&c^{t+i}_{\eta t+\ell t+1}+e_{\eta+\ell+1}c^{i}_{\eta t+\ell t+2},c^{t+i}_{\eta t+\ell t+2},\ldots,c^{t+i}_{\eta t+\ell t+t}+c^{(t-1)t+i}_{\eta t+\ell t+2},\\
& \vdots,\\
&c^{(t-1)t+i}_{\eta t+\ell t+1}+e_{\eta+\ell+1}c^{i}_{\eta t+\ell t+t},c^{(t-1)t+i}_{\eta t+\ell t+2}+e_{\eta +\ell+1}c^{t+i}_{\eta t+\ell t+t},\ldots,c^{(t-1)t+i}_{\eta t+\ell t+t},
\end{align*}
where $\ell=0,1,\ldots,\eta -1$.
 By Lemma \ref{lm:trans}, we can compute $c^{i}_{\eta t+\ell t+2}$ and $c^{t+i}_{\eta t+\ell t+1}$ from $c^{i}_{\eta t+\ell t+2}+c^{t+i}_{\eta t+\ell t+1}$ and $c^{t+i}_{\eta t+\ell t+1}+e_{\eta+\ell+1}c^{i}_{\eta t+\ell t+2}$. Similarly,  by Lemma \ref{lm:trans}, we can compute all the symbols in \eqref{eq:thm-pr1} from the above symbols. Therefore, according to the MDS property, we can obtain $kt$ symbols in \eqref{eq:thm-pr} and further compute
$c^{i+\ell t}_{1},c^{i+\ell t}_{2},\ldots,c^{i+\ell t}_{t}$ with $\ell=0,1,\ldots,t-1$. With the computed $c^{i+\ell t}_{1},c^{i+\ell t}_{2},\ldots,c^{i+\ell t}_{t}$ with $\ell=0,1,\ldots,t-1$ and the accessed $(d-k)t$ symbols in~\eqref{eq:t(d-k)-symbol},  by
Lemma \ref{lm:trans}, we can compute all $t$ symbols stored in node $(j-1)t+i$. Therefore, we can recover $t^2$ symbols in node $(j-1)t+i$ by downloading $td$ symbols from $d$ helper nodes and the repair access is optimal.
\end{proof}

According to Theorem \ref{thm:opt-repair} and Theorem \ref{thm:repair-preserve},
the codes $\mathcal{C}_{2}(n,k,\eta,t)$ have optimal repair access for each of the first $2\eta t$ nodes.
By recursively applying the transformation for set $i+1$ of the codes $\mathcal{C}_{i}(n,k,\eta ,t)$
for $i=1,2,\ldots,\lceil \frac{n}{\eta t}\rceil-1$, we can obtain the codes
$\mathcal{C}_{\lceil \frac{n}{\eta t}\rceil}(n,k,\eta,t)$ that satisfy the MDS property according to
Theorem \ref{thm:hybridcons} and have optimal repair access for any single node according to
Theorem \ref{thm:opt-repair} and Theorem \ref{thm:repair-preserve}.

\subsection{Repair Method}
For $i=1,2,\ldots,t$, $j=1,2,\ldots,\eta$ and $\ell=1,2,\ldots,\lceil \frac{n}{\eta t}\rceil$,
we can repair $t^{\lceil \frac{n}{\eta t}\rceil}$ symbols in node $(\ell-1)\eta t+(j-1)t+i$ by downloading
the symbols in row $f$ for
$$f\bmod (t^\ell)\in \{(i-1)(t^{\ell-1})+1,(i-1)(t^{\ell-1})+2,\ldots,i(t^{\ell-1})-1\}$$
from $\eta-1$ nodes $(\ell-1)\eta t+i,\ldots,(\ell-1)\eta t+(j-2)t+i,(\ell-1)\eta t+jt+i,\ldots,(\ell-1)\eta t+(\eta-1)t+i$, $t-1$ nodes $(j-1)t+1,\ldots,(j-1)t+i-1,(j-1)t+i+1,\ldots,(j-1)t+t$,
and $d-\eta-t+2$ nodes $h_1,h_2,\ldots,h_{d-\eta-t+2}$. For $i=1,2,\ldots,d-\eta-t+2$, if $h_i$ belongs to a group in set $\mu$ with $\mu>\ell$, then all $t$ nodes of the group should be chosen in the helper nodes $h_1,h_2,\ldots,h_{d-\eta-t+2}$. With the same argument of the proof of Theorem \ref{thm:repair-preserve}, we can show that node $(\ell-1)\eta t+(j-1)t+i$ can be repaired by the above repair method with repair access being $d\cdot t^{\lceil \frac{n}{\eta t}\rceil-1}$, which is,  according to \eqref{eq:optimal-repair}, optimal.

The example given in Section \ref{sec:sec-exam} is the code $\mathcal{C}_{2}(n,k,\eta ,t)$ with $n=8$, $k=5$, $d=6$ and $\eta=2$, which is shown in Table \ref{table:A3}.
Suppose that node 1 fails, i.e., $i=1$, $j=1$ and $\ell=1$. According to the above repair method, we can repair node 1 by downloading the symbols in row $f$ for $f\bmod 2\in \{1\}$ from nodes $2,3,5,6,7,8$, i.e.,
\begin{align*}
&c^1_2+c^2_1,c^1_3,c^1_5,c^1_6+c^3_5,c^1_7,c^1_8+c^3_7,\\
&c^3_2+c^4_1,c^3_3,c^3_5+e_3c^1_6,c^3_6,c^3_7+e_4c^1_8,c^3_8.
\end{align*}
The detailed repair procedure of node 1 is illustrated in Fig. \ref{example}.
When $i=1$, $j=2$ and $\ell=1$, we can repair node 3 by downloading the symbols in rows 1 and 3 from nodes $1,4,5,6,7,8$. When $i=2$, $j=1$ and $\ell=1$, we can repair node 2 by downloading the symbols in rows 2 and 4 from nodes $1,4,5,6,7,8$. When $i=2$, $j=2$ and $\ell=1$, we can repair node 2 by downloading the symbols in rows 2 and 4 from nodes $2,3,5,6,7,8$.
%Node 2 and node 4 can be repaired by downloading the symbols in rows $2,4$ from nodes $1,4,5,6,7,8$ and $2,3,5,6,7,8$, respectively.
According to the above repair method, we can repair node 5 and node 7 by downloading the symbols in rows $1,2$ from nodes $1,2,3,4,6,7$ and $1,2,3,4,5,8$, respectively, and repair node 6 and node 8 by downloading the symbols in rows $3,4$ from nodes $1,2,3,4,5,8$ and $1,2,3,4,6,7$, respectively.

\section{Multi-Layer Transformed Binary MDS Array Codes}
\label{sec:trans-binary}
In this section, we present the method of applying the proposed transformation for binary
MDS array codes, using EVENODD codes as a motivating example, to obtain the multi-layer
transformed EVENODD codes that have optimal repair access.

\subsection{EVENODD Codes}
An EVENODD code can be presented by a $(p-1)\times (k+r)$ array $[a_{i,j}]$
for $i=0,1,\ldots,p-2$ and $j=0,1,\ldots,k+r-1$. The first $k$ columns are information columns
that store information bits and the last $r$ columns are parity columns that store parity bits.
For $j=0,1,\ldots,k+r-1$, we represent the $p-1$ bits $a_{0,j},a_{1,j},\ldots,a_{p-2,j}$ in column $j$
by the polynomial
\[
a_j(x)=a_{0,j}+a_{1,j}x+\ldots+a_{p-2,j}x^{p-2}.
\]
The first $k$ polynomials $a_0(x),\ldots,a_{k-1}(x)$ are information polynomials, and the last
$r$ polynomials $a_{k}(x),\ldots,a_{k+r-1}(x)$ are parity polynomials that are computed as
\[\begin{bmatrix}
a_{k}(x)& \cdots & a_{k+r-1}(x)\end{bmatrix}
=\begin{bmatrix}a_{0}(x)& \cdots & a_{k-1}(x)\end{bmatrix}\cdot
\begin{bmatrix}
1& 1& \cdots &1\\
1& x& \cdots &x^{r-1}\\
\vdots & \vdots & \ddots & \vdots \\
1& x^{k-1} & \cdots & x^{(r-1)(k-1)}
\end{bmatrix}
\]
over the ring $\mathbb{F}_2[x]/(1+x+\cdots+x^{p-1})$. Given parameters $k$ and $r$, we need to choose
the parameter $p$ such that the MDS property is satisfied.
The MDS property condition of EVENODD codes is given in \cite{blaum1996,hou2016on}.

\subsection{Multi-Layer Transformed EVENODD Codes}
\label{sec:evenodd}
We first present how to apply the proposed transformation in Section \ref{sec:trans} for EVENODD
codes such that the transformed EVENODD codes have optimal repair access for each of the
chosen $(d-k+1)\eta$ columns, and then give the multi-layer transformed EVENODD codes with optimal repair access for any single column as in Section \ref{sec:const}.

The transformed EVENODD code is an array code of size $(p-1)(d-k+1)\times (k+r)$. Given the
$(p-1)(d-k+1)\times k$ information array
\[
\begin{bmatrix}
a^1_{0,1}& a^1_{0,2}& \cdots &a^1_{0,k}\\
a^1_{1,1}& a^1_{1,2}& \cdots &a^1_{1,k}\\
\vdots & \vdots & \ddots & \vdots \\
a^1_{p-2,1}& a^1_{p-2,2}& \cdots &a^1_{p-2,k}\\ \hline
\vdots & \vdots & \ddots & \vdots \\ \hline
a^{d-k+1}_{0,1}& a^{d-k+1}_{0,2}& \cdots &a^{d-k+1}_{0,k}\\
a^{d-k+1}_{1,1}& a^{d-k+1}_{1,2}& \cdots &a^{d-k+1}_{1,k}\\
\vdots & \vdots & \ddots & \vdots \\
a^{d-k+1}_{p-2,1}& a^{d-k+1}_{p-2,2}& \cdots &a^{d-k+1}_{p-2,k}\\
\end{bmatrix},
\]
we first represent each $p-1$ information bits $a^{\ell}_{0,j},a^{\ell}_{1,j},\ldots,a^{\ell}_{p-2,j}$ by the
{\em information polynomial}
\[
a^\ell_j(x)=a^\ell_{0,j}+a^\ell_{1,j}x+\ldots+a^\ell_{p-2,j}x^{p-2},
\]
and then compute $(d-k+1)n$ {\em coded polynomials} by
\[\begin{bmatrix}
a^\ell_{k+1}(x)& \cdots & a^\ell_{k+r}(x)\end{bmatrix}
=\begin{bmatrix}a^\ell_{1}(x)& \cdots & a^\ell_{k}(x)\end{bmatrix}\cdot
\begin{bmatrix}
1& 1& \cdots &1\\
1& x& \cdots &x^{r-1}\\
\vdots & \vdots & \ddots & \vdots \\
1& x^{k-1} & \cdots & x^{(r-1)(k-1)}
\end{bmatrix}
\]
over the ring $\mathbb{F}_2[x]/(1+x+\cdots+x^{p-1})$, where $\ell=1,2,\ldots,d-k+1$ and $j=1,2,\ldots,k$.
Recall that $t=d-k+1$ and $\eta=\lfloor \frac{r-1}{d-k}\rfloor$.
For $i=1,2,\ldots, t$ and $j=1,2,\ldots,\eta$, column $(j-1)t+i$ stores the following $t$
polynomials
\begin{equation}
\begin{array}{ll}
& a^1_{(j-1)t+i}(x)+a^i_{(j-1)t+1}(x),\\
& a^2_{(j-1)t+i}(x)+a^i_{(j-1)t+2}(x),\ldots,\\
& a^{i-1}_{(j-1)t+i}(x)+a^i_{(j-1)t+i-1}(x),\\
& a^i_{(j-1)t+i}(x),\\
& a^{i+1}_{(j-1)t+i}(x)+e_j(x)a^i_{(j-1)t+i+1}(x),\\
& a^{i+2}_{(j-1)t+i}(x)+e_j(x)a^i_{(j-1)t+i+2}(x),\ldots,\\
& a^{t}_{(j-1)t+i}(x)+e_j(x)a^i_{(j-1)t+t}(x),
\end{array}
\label{eq:trans-symbol1}
\end{equation}
where $e_j(x)$ is a non-zero polynomial in $\mathbb{F}_2[x]/(1+x+\cdots+x^{p-1})$ such
that both $e_j(x)$ and $e_j(x)+1$ are invertible in $\mathbb{F}_2[x]/(1+x+\cdots+x^{p-1})$.
For $h=t\cdot \eta+1, t\cdot \eta+2,\ldots,n$, column $h$ stores $t$ polynomials
\[
a^1_{h}(x),a^2_{h}(x),\ldots,a^{t}_{h}(x).
\]
The above obtained codes are called {\em transformed EVENODD codes}.
The transformation
in \cite{hou2018} can be viewed as a special case of our transformation with
$\eta=1$. Note that the
transformed codes in Section \ref{sec:trans} and the transformed EVENODD codes
are essentially the same codes, the difference is that the transformed codes in
Section \ref{sec:trans} are operated over the finite field $\mathbb{F}_q$ and the
transformed EVENODD codes are operated over the ring $\mathbb{F}_2[x]/(1+x+\cdots+x^{p-1})$.

With the same proof of Theorem \ref{thm:hybridcons}, we can show that the transformed
EVENODD codes satisfy the MDS property if EVENODD codes satisfy the MDS property
and $p$ is large enough. We can also show that the repair access of each of the first
$\eta t$ columns is optimal, as like the transformed codes in Section \ref{sec:trans}.
By recursively applying the transformation for EVENODD codes $\lceil \frac{n}{\eta t}\rceil$ times,
we can obtain the multi-layer transformed $\mathsf{EVENODD}_{\lceil \frac{n}{\eta t}\rceil}(n,k,\eta,t)$ codes
that have optimal repair access for any single column and satisfy the MDS property when $p$
is large enough.

\subsection{Transformation for Other Binary MDS Array Codes}
The transformation can also be employed in other binary MDS array codes, such as
the codes in \cite{blomer1999,feng2005,hou2018a,hou2017,hou2018c,hou2018b}, to enable optimal repair access for any single column.

Specifically, the transformation for RDP and codes in \cite{blomer1999,feng2005,hou2018a} is similar to the transformation for EVENODD
codes in Section \ref{sec:evenodd}. By applying the transformation for RDP $\lceil \frac{n}{\eta t}\rceil$ times,
we can show that the obtained multi-layer transformed $\mathsf{RDP}_{\lceil \frac{n}{\eta t}\rceil}(n,k,\eta,t)$ codes
have optimal repair access for any single column, as in $\mathsf{EVENODD}_{\lceil \frac{n}{\eta t}\rceil}(n,k,\eta,t)$
codes.

Recall that the codes in \cite{hou2017,hou2018c,hou2018b,gad2013,pamies2016} have optimal repair access or efficient optimal repair
access for information column, we only need to apply the transformation for them
$\lceil \frac{r}{\eta t}\rceil$ times to obtain the transformed codes with lower
sub-packetization level that have optimal repair access
for any parity column and optimal repair access or efficient repair access for any information column.
Note that to preserve the efficient repair property of the information column, we need to
carefully design the transformation for different codes in \cite{hou2017,hou2018c,hou2018b,gad2013,pamies2016}, and that will be one of our future work.

\section{Comparison}
\label{sec:com}
\begin{table*}[!t]
%\scriptsize
\caption{Comparison of existing MDS codes with optimal repair access.}
%\vspace{-12pt}
\begin{center}
\begin{tabular}{|c|c|c|c|}
\hline
MDS codes &  No. of helpers $d$  & Sub-packetization $\alpha$ & No. of nodes with optimal repair  \\
\hline
The first codes in \cite{ye2017}& $n-1$ & $r^{n}$ & $n$ \\
\hline
The second codes in \cite{ye2017}& $n-1$ & $r^{n-1}$ & $n$ \\
\hline
Codes in \cite{ye2017explicit1}& any $d$ & $(d-k+1)^{\lceil \frac{n}{d-k+1}\rceil}$ & $n$ \\
\hline
Codes in \cite{myna2018}& any $d$ & $(d-k+1)^{\lceil \frac{n}{d-k+1}\rceil}$ & $n$ \\
\hline
Codes in \cite{li2017}& $n-1$ & $r^{\lceil \frac{n}{r}\rceil}$ & $n$ \\
\hline
\cite[Corollary 4]{balaji2017}& $n-1$ & $\geq r^{\lceil \frac{n}{r}\rceil}$ & $n$ \\
\hline
\cite[Corollary 6]{balaji2017}& any $d$ & $\geq (d-k+1)^{\lceil \frac{n-1}{d-k+1}\rceil}$ & $n$ \\
\hline
\cite[Corollary 5]{balaji2017}& any $d$ & $\geq (d-k+1)^{\lceil \frac{w}{d-k+1}\rceil}$ & $k-1<w \leq n$ \\
\hline
Proposed $\mathcal{C}_{\lceil \frac{n}{\eta t}\rceil}(n,k,\eta,t)$& any $d$ & $(d-k+1)^{\lceil \frac{n}{\eta(d-k+1)}\rceil}$ & $n$ \\
\hline
Proposed  $\mathcal{C}_{\lceil \frac{n}{2\eta }\rceil}(n,k,\eta,t=2)$& $d=k+1$ & $2^{\lceil \frac{n}{2\eta}\rceil}$ & $n$ \\
\hline
Proposed  $\mathcal{C}_{\lceil \frac{w}{\eta t}\rceil}(n,k,\eta,t)$& any $d$ & $(d-k+1)^{\lceil \frac{w}{\eta(d-k+1)}\rceil}$ & $w$ \\
\hline
Proposed  $\mathcal{C}_{\lceil \frac{w}{2\eta}\rceil}(n,k,\eta,t=2)$& $d=k+1$ & $2^{\lceil \frac{w}{2\eta}\rceil}$ & $k-1 <w\leq n$ \\
\hline
\end{tabular}
\end{center}
\label{table:comparison}
\end{table*}

A comparison of existing MDS codes with optimal repair access and the proposed codes $\mathcal{C}_{\lceil \frac{n}{\eta t}\rceil}(n,k,\eta,t)$ is given in Table \ref{table:comparison}. The results in Table \ref{table:comparison} show that the proposed codes $\mathcal{C}_{\lceil \frac{n}{\eta t}\rceil}(n,k,\eta,t)$ have two advantages: $(i)$ when $\lfloor \frac{r-1}{d-k}\rfloor >1$, the proposed codes $\mathcal{C}_{\lceil \frac{n}{\eta t}\rceil}(n,k,\eta,t)$ with optimal repair access for any single node have lower sub-packetization level compared to both the existing MDS codes in \cite{ye2017explicit1,myna2018} with optimal repair access for any single node and the lower bound of sub-packetization level of MDS codes with optimal repair access for any single node in \cite[Corollary 6]{balaji2017}; $(ii)$ when $\lfloor \frac{r-1}{d-k}\rfloor >1$, the proposed codes $\mathcal{C}_{\lceil \frac{w}{\eta t}\rceil}(n,k,\eta,t)$ with optimal repair access for each of $w$ nodes have less sub-packetization level than the tight lower bound on the sub-packetization level of MDS codes with optimal repair access for each of $w$ nodes in \cite[Corollary 5]{balaji2017}. Table \ref{table:comparison1} shows the sub-packetization level of the proposed codes and the codes in
\cite{ye2017explicit1,myna2018} for some parameters.

\begin{table*}[!t]
%\scriptsize
\caption{Sub-packetization level of the proposed codes and the codes in
\cite{ye2017explicit1,myna2018} for some parameters.}
%\vspace{-12pt}
\begin{center}
\begin{tabular}{|c|c|c|c|c|c|c|}
\hline
Parameters $(n,k,d)$ & $(14,10,11)$  & $(12,8,9)$ & $(18,14,15)$ & $(18,13,15)$ & $(24,19,21)$ & $(80,71,72)$  \\
\hline
$\alpha$ in \cite{ye2017explicit1,myna2018}& $128$ & $64$ & $512$& $729$ & $6561$& $2^{40}=1 099 511 627 776$\\
\hline
$\alpha$ of the proposed codes& $8$ & $4$ & $8$& $27$ & $81$& $1024$\\
\hline
\end{tabular}
\end{center}
\label{table:comparison1}
\end{table*}

Compared with the existing binary MDS array codes with optimal repair access for any column, the
proposed $\mathsf{EVENODD}_{\lceil \frac{n}{\eta t}\rceil}(n,k,\eta,t)$ codes have lower sub-packetization level.
Recall that the sub-packetization level of both the transformed EVENODD codes in \cite{hou2018} and
the transformed codes in \cite{li2019} is
$(d-k+1)^{\lceil \frac{n}{ d-k+1}\rceil}$, which is strictly larger than the sub-packetization level of $\mathsf{EVENODD}_{\lceil \frac{n}{\eta t}\rceil}(n,k,\eta,t)$ when $\lceil \frac{r-1}{d-k}\rceil \geq 2$.

\section{Conclusions and Future Work}
\label{sec:conclu}

In this paper, we first propose a  transformation for MDS codes that can enable optimal repair
access for each of the chosen $(d-k+1)\eta$ nodes. By applying the proposed transformation for any
MDS code $\lceil \frac{n}{\eta t}\rceil$ times, we can obtain a multi-layer transformed code that
has optimal repair access for any single node. With a slightly modification, we can also design the transformation
for binary MDS array codes where we  use EVENODD codes as a motivating example. We show that
the $\mathsf{EVENODD}_{\lceil \frac{n}{\eta t}\rceil}(n,k,\eta,t)$ codes
obtained by applying the transformation for EVENODD codes $\lceil \frac{n}{\eta t}\rceil$ times have
optimal repair access for any single column. Moreover, the proposed multi-layer transformed codes have lower
sub-packetization level than that of the existing MDS codes with optimal repair access for any single node.
How to design a specific transformation for binary MDS codes in \cite{hou2017,hou2018c,hou2018b,gad2013,pamies2016} with
efficient repair access is one of our
future work. When more than one column fail, how to retrieve all  information bits from any
$k$ of the surviving nodes with lower computational complexity is a challenge decoding problem. How to
design MDS codes with lower sub-packetization level, efficient repair access for any single node and
lower decoding complexity is another future work.

\appendices

\ifCLASSOPTIONcaptionsoff
  \newpage
\fi

\bibliographystyle{IEEEtran}

\end{document}